\documentclass[nohyperref]{article}

\usepackage{microtype}
\usepackage{graphicx}
\usepackage{subfigure}
\usepackage{booktabs} %
\usepackage{listings}

\usepackage{hyperref}

\usepackage[accepted]{icml/icml2023}

\usepackage{amsmath}
\usepackage{amssymb}
\usepackage{mathtools}
\usepackage{amsthm}

\usepackage{color,soul}

\usepackage[capitalize,noabbrev]{cleveref}

\theoremstyle{plain}
\newtheorem{theorem}{Theorem}[section]
\newtheorem{proposition}[theorem]{Proposition}
\newtheorem{lemma}[theorem]{Lemma}
\newtheorem{corollary}[theorem]{Corollary}
\theoremstyle{definition}

\theoremstyle{remark}

\usepackage[textsize=tiny]{todonotes}

\graphicspath{{figures/},{figures/covcomponents/}}
\DeclareGraphicsExtensions{.pdf,.eps,.png,.jpg,.jpeg}

\newcommand{\bx}{{\mathbf{x}}}
\newcommand{\by}{{\mathbf{y}}}

\newcommand{\bn}{{\boldsymbol{n}}}

\newcommand{\zeros}{{\mathbf{0}}}
\newcommand{\eye}{{\mathbf{I}}}

\newcommand{\btheta}{{\boldsymbol{\theta}}}
\newcommand{\bSigma}{{\boldsymbol{\Sigma}}}
\newcommand{\beps}{{\boldsymbol{\epsilon}}}
\newcommand{\bGamma}{{\boldsymbol{\Gamma}}}

\newcommand{\cov}{\mathbb{C}\mathrm{ov}}
\newcommand{\tr}{\mathrm{Tr}}
\DeclareMathOperator{\Exp}{Exp}
\DeclareMathOperator{\Log}{Log}

\newcommand{\E}{\mathbb{E}}
\newcommand{\Cov}{\mathbb{C}\mathrm{ov}}
\newcommand{\balpha}{\boldsymbol{\alpha}}
\newcommand{\R}{\mathbb{R}}
\newcommand{\bbS}{\mathbb{S}}
\newcommand{\bfn}{\boldsymbol{n}}

\newcommand{\norm}[1]{\left\lVert #1 \right\rVert}
\newcommand{\inner}[2]{\left\langle #1,\ #2 \right\rangle}
\DeclareMathOperator*{\argmin}{arg\,min}

\icmltitlerunning{\hfill Multi-Fidelity Covariance Estimation~\hfill \thepage}

\begin{document}

\twocolumn[
\icmltitle{Multi-Fidelity Covariance Estimation in the Log-Euclidean Geometry}

\icmlsetsymbol{equal}{*}

\begin{icmlauthorlist}
\icmlauthor{Aimee Maurais}{MIT}
\icmlauthor{Terrence Alsup}{NYU}
\icmlauthor{Benjamin Peherstorfer}{NYU}
\icmlauthor{Youssef Marzouk}{MIT}
\end{icmlauthorlist}

\icmlaffiliation{MIT}{Massachusetts Institute of Technology, Cambridge, MA USA}
\icmlaffiliation{NYU}{Courant Institute of Mathematical Sciences, New York University, New York, NY USA}

\icmlcorrespondingauthor{Aimee Maurais}{maurais@mit.edu}
\icmlcorrespondingauthor{Terrence Alsup}{alsup@cims.nyu.edu}

\icmlkeywords{Monte Carlo methods, geometric learning, matrix computations, surrogate modeling, physical sciences, covariance estimation} %

\vskip 0.3in
]

\printAffiliationsAndNotice{}  %

\begin{abstract}
We introduce a multi-fidelity estimator of covariance matrices that employs the log-Euclidean geometry of the symmetric positive-definite manifold. The estimator fuses samples from a hierarchy of data sources of differing fidelities and costs for variance reduction while guaranteeing definiteness, in contrast with previous approaches. The new estimator makes covariance estimation tractable in applications where simulation or data collection is expensive; to that end, we develop an optimal sample allocation scheme that minimizes the mean-squared error of the estimator given a fixed budget. Guaranteed definiteness is crucial to metric learning, data assimilation, and other downstream tasks. Evaluations of our approach using data from physical applications (heat conduction, fluid dynamics) demonstrate more accurate metric learning and speedups of more than one order of magnitude compared to benchmarks. %
\end{abstract}

\section{Introduction}
\label{sec:intro}
Covariance estimation is a fundamental task for understanding the relationships between multiple features of data points. %
It arises in a wide range of machine learning applications such as metric learning \citep{Zadeh2016,pmlr-v37-huanga15}, graphical models \citep{JMLR:v15:loh14a,10.1214/11-EJS631}, and classification \citep{hastie01statisticallearning}, as well as in spatial statistics \citep{cressie2015statistics} and in science and engineering applications such as filtering and data assimilation \citep{doi:10.1137/16M105959X,https://doi.org/10.1002/qj.591,LawBook}. 
The canonical Monte Carlo approach to estimating the covariance matrix $\bSigma = \cov[\by]$ of a random variable $\by$ with $d$ components is to take $n$ samples $\by_1, \dots, \by_n$ %
and compute the sample covariance matrix
    \begin{equation}
        \hat{\bSigma}_n = \frac{1}{n}\sum\nolimits_{i=1}^n \left(\by_i - \bar{\by}\right)\left(\by_i - \bar{\by}\right)^{\top} \, ,
    \label{eq:SampleCov}
    \end{equation}
    where\footnote{If the mean of $\by$ is unknown, one can replace $\bar{\by}$ with the sample mean $\frac{1}{n}\sum\nolimits_{i=1}^n \by_i$ and normalize \eqref{eq:SampleCov} with $1/(n-1)$.} $\bar{\by} = \mathbb{E}[\by]$.
    The mean-squared error (MSE) of the estimator $\hat{\bSigma}_n$ in the Euclidean metric, i.e., $\mathbb{E}[d_{\mathrm{F}}(\hat{\bSigma}, \bSigma)^2]$ where $d_{\mathrm{F}}({\mathbf{A}}, {\mathbf{B}}) = \norm{ {\mathbf{A}} - {\mathbf{B}}}_{\mathrm{F}}$ and $\norm{\cdot}_{\mathrm{F}}$ denotes the Frobenius norm,
    decays as $\mathcal{O}(1/n)$ %
    if the samples are independent and identically distributed (i.i.d.). %
Furthermore, if $\bSigma$ lies in the space $\mathbb{S}_{++}^d$ of symmetric (strictly) positive definite (SPD) $d\times d$ matrices, then the sample covariance matrix $\hat{\bSigma}_n$ is in $\mathbb{S}_{++}^d$ almost surely for $n > d$. 
We are interested in situations where obtaining a realization of $\by$ incurs a high cost $c_0 > 0$ and thus straightforward application of the Monte Carlo estimator \eqref{eq:SampleCov} becomes computationally intractable. High sampling cost is a pervasive issue in science and engineering, where sampling typically corresponds to measuring the response of processes and systems that must be either numerically simulated or physically observed. We note that there exist a wide range of techniques for covariance estimation in high dimensions \citep{10.1214/009053607000000758,10.1214/12-AOS989,NIPS2012_6ba1085b,10.1214/11-AOS944,6584826,JMLR:v15:loh14a,10.1214/17-AOS1601}. Our focus here instead is on situations where \emph{drawing each sample incurs high cost}; see also the discussion in the Conclusions. 

\textbf{Multi-fidelity estimation}
The opportunity that we exploit is that in many applications 
there is a \emph{multi-fidelity hierarchy} of data sources from which correlated auxiliary samples can be obtained; see \citet{PeherstorferEtAl2018} for a survey. Multi-fidelity approaches have found increasing utility in machine learning---e.g., for
optimization \cite{NIPS2017_df1f1d20,pmlr-v115-wu20a,NEURIPS2020_60e1deb0}, classification and sequential learning \cite{NEURIPS2018_01a06836,pmlr-v161-gundersen21a,Palizhati2022}, and sampling \cite{https://doi.org/10.48550/arxiv.2210.01534,doi:10.1137/18M1229742,pmlr-v145-alsup22a}. %
Let $(\Omega, \mathcal{F}, \mathbb{P})$ be a probability triple. 
In the multi-fidelity setting, in addition to the original ``high-fidelity'' random variable $\by \equiv \by^{(0)}: \Omega \to \R^d$ generating the sample set $\{\by^{(0)}(\omega_i)\}_{i = 1}^{n_0}$, one has access to a number of low-fidelity surrogate random variables $\by^{(1)}, \dots, \by^{(\ell)}: \Omega \to \R^d$ and  obtains corresponding sets of samples $\{\by^{(\ell)}(\omega_i)\}_{i = 1}^{n_{\ell}}$, $\ell = 1, \dots, L$. For a particular $\omega \in \Omega$, the realizations $\by^{(0)}(\omega),\, \by^{(1)}(\omega), \dots, \by^{(L)}(\omega)$ are correlated, but the realizations within each set $\{\by^{(\ell)}(\omega_i)\}_{i = 1}^{n_{\ell}}$ are i.i.d., for $\ell = 0, \dots, L$. We assume that the surrogates $\by^{(1)}, \dots, \by^{(L)}$ are sorted in order of decreasing fidelity to $\by^{(0)}$ as measured by correlation coefficients $1 \geq |\rho_{1}| \geq |\rho_2| \geq \cdots \geq |\rho_L|$. Here $\rho_\ell$ is a notion of multivariate correlation between $\by^{(0)}(\omega)$ and $\by^{(\ell)}(\omega)$, $\ell = 1, \dots, L$, that will be formally defined in Section~\ref{sec:3.3}. Corresponding to decreasing fidelity, the surrogate random variables have decreasing sampling costs, $c_0 \geq c_1  \geq \dots \geq c_L$. %
    
\textbf{Loss of definiteness in Euclidean multi-fidelity covariance estimation} Direct application of multi-level and multi-fidelity Monte Carlo estimation \citep{Gi08,cliffe_multilevel_2011,Teckentrup2013,NME:NME4761} based on control variates to (co)variance estimation has been proposed in, e.g., \citet{Bierig2015,QianEtAl2018,MycekDeLozzo2019}. These multi-fidelity estimators, however, rely on differences of single-fidelity Monte Carlo estimators, which can lead to a loss of definiteness of the estimated covariance matrix, as we will detail in Section~\ref{sec:MFSetup}. Ignoring this loss-of-definiteness can result in errors in downstream tasks (e.g., instabilities in Kalman filtering, distances becoming negative in a learned metric, inadmissible graph structures). Accordingly, post-processing strategies \citep{HoelEtAl2016,Chernov2021} have been proposed that eliminate negative eigenvalues, but elimination of eigenvalues introduces potential for large errors, requires hand-tuning of thresholds, and often, as the authors note, incurs high computational cost; see also \citet{doi:10.1137/050624509,PopovEtAl2021}. %

\textbf{Our contribution:  Guaranteed positive multi-fidelity log-Euclidean covariance estimation}
We introduce a multi-fidelity covariance estimator based on control variates in the log-Euclidean geometry for $\mathbb{S}^d_{++}$ \citep{ArsignyEtAl2006}. The log-Euclidean geometry equips $\mathbb{S}^d_{++}$ with vector-space structure by placing it in one-to-one correspondence with $\mathbb{S}^d$, the vector space of $d \times d$ symmetric matrices, and defining notions of logarithmic addition and scalar multiplication for elements of $\mathbb{S}^{d}_{++}$. In particular, we make use of this vector space structure to safely ``subtract'' sample covariance matrices as part of a control variate construction in the tangent space to $\mathbb{S}_{++}^d$, which we identify with $\mathbb{S}^d$. The log-Euclidean vector space can be equipped with the log-Euclidean metric for SPD matrices, which itself induces Riemannian structure on $\mathbb{S}_{++}^d$. An additional advantage of the log-Euclidean geometry used here, as compared to many other geometries for $\bbS^d_{++}$ (see Section~\ref{sec:LEMF:Geometry}), is that computations remain efficient as the most complex operations required are matrix exponentials and logarithms.  %
    
\textbf{Summary of main contributions} The main contributions and key features of this work are summarized as follows:\\
\hspace*{0.2cm} \textbf{(a)} We introduce a multi-fidelity estimator that preserves definiteness of covariance matrices in the finite-sample regime and that has first-order minimal mean squared error (MSE) in the log-Euclidean metric.\\
\hspace*{0.2cm} \textbf{(b)} We provide analysis that leads to a first-order optimal sample allocation between the high-fidelity and surrogate models for gaining orders of magnitude speedups compared to using high-fidelity samples alone.\\
\hspace*{0.2cm} \textbf{(c)} We demonstrate the algorithmic benefits of the log-Euclidean geometry by showing that our estimator can be implemented with few lines of code (Appendix~\ref{sec:Appdx:Code}), relying on standard numerical linear algebra routines. This low barrier to implementation allows the estimator to be readily plugged into existing code with minimal effort and even as a post-processing step.   

\section{Multi-Fidelity Covariance Estimation in Euclidean Geometry and Loss of Definiteness} \label{sec:MFSetup}
Recall from the introduction that we generate a set of high-fidelity samples $\{\by^{(0)}(\omega_i)\}_{i = 1}^{n_0}$ and sets of surrogate samples $\{\by^{(\ell)}(\omega_i)\}_{i = 1}^{n_{\ell}}$ over different fidelity levels $\ell = 1, \dots, L$. The classical approach to multi-fidelity estimation, i.e., in the Euclidean geometry, is to use the surrogate samples to define control variates that reduce the variance of the standard Monte Carlo estimator \eqref{eq:SampleCov} \citep{Bierig2015,QianEtAl2018,MycekDeLozzo2019}. This approach yields the Euclidean multi-fidelity (EMF) covariance estimator
\begin{equation}
    \hat{\bSigma}_{\bn}^{\mathrm{EMF}}
    =
    \hat{\bSigma}_{n_0}^{(0)} + \sum\nolimits_{\ell=1}^L \alpha_{\ell} \left( \hat{\bSigma}_{n_{\ell}}^{(\ell)} - \hat{\bSigma}^{(\ell)}_{n_{\ell-1}} \right)  \, ,
    \label{eq:LCV}
    \end{equation}
    where $\balpha = [\alpha_1,\ldots,\alpha_L]^{\top} \in \mathbb{R}^L$ are the control variate weights
    and %
    \begin{equation}
    \begin{aligned}
    \hat{\bSigma}^{(\ell)}_{n_{\ell}} = &\frac{1}{n_\ell}\sum_{i=1}^{n_\ell} \left(\by^{(\ell)}(\omega_i) - \bar{\by}^{(\ell)}_{n_\ell} \right)\left(\by^{(\ell)}(\omega_i) - \bar{\by}^{(\ell)}_{n_\ell}\right)^{\top}, \\
    \hat{\bSigma}^{(\ell)}_{n_{\ell - 1}} = &\frac{1}{n_{\ell - 1}}\sum_{i=1}^{n_{\ell - 1}} \left(\by^{(\ell)}(\omega_i) - \bar{\by}^{(\ell)}_{n_\ell} \right)\left(\by^{(\ell)}(\omega_i) - \bar{\by}^{(\ell)}_{n_\ell}\right)^{\top}
    \end{aligned}
    \label{eq:Euc:SampleCov}
    \end{equation}
    are sample covariance matrices computed from $\{\by^{(\ell)}(\omega_i)\}_{i=1}^{n_\ell}$ and $\{\by^{(\ell)}(\omega_i)\}_{i=1}^{n_{\ell - 1}}$, $\ell = 1, \dots, L$. We require for well-definedness that $n_{\ell} \geq n_{\ell - 1}$, $\ell = 1, \dots, L$. 
    For levels $\ell, t \in \{0, \dots, L\}$ and integers $m_i, m_j \in \mathbb{Z}^+$, notice that the sample covariance matrices 
    \begin{gather*}
    \hat\bSigma_{m_i}^{(\ell)} = \frac{1}{m_i}\sum_{i'=1}^{m_i} \left(\by^{(\ell)}(\omega_{i'}) - \bar{\by}^{(\ell)}_{m_i} \right)\left(\by^{(\ell)}(\omega_{i'}) - \bar{\by}^{(\ell)}_{m_i}\right)^{\top}\,, \\
    \hat\bSigma_{m_j}^{(t)} = \frac{1}{m_j}\sum_{j'=1}^{m_j} \left(\by^{(t)}(\omega_{j'}) - \bar{\by}^{(t)}_{m_j} \right)\left(\by^{(t)}(\omega_{j'}) - \bar{\by}^{(t)}_{m_j}\right)^{\top}
    \end{gather*}
    are correlated because they are constructed from evaluations of $\by^{(\ell)}$ and $\by^{(t)}$ on the common set of events $\{\omega_i\}_{i=1}^{\min\{m_i, m_j\}}$. This correlation is the key ingredient to achieving variance reduction in \eqref{eq:LCV} over equivalent-cost single-fidelity estimators; see \cref{sec:3.3}.  
    
\textbf{Loss of definiteness} The use of regular subtraction and scalar multiplication in \cref{eq:LCV} treats the symmetric positive definite matrices $\hat\bSigma^{(\ell)}_{n_\ell}$ and $\hat\bSigma_{n_{\ell-1}}^{(\ell)}$, $\ell = 0, \dots, L$, as elements of the Euclidean vector space $\mathbb{S}^d$. However, SPD matrices do not constitute a vector space under the Euclidean geometry and thus the Euclidean multi-fidelity estimator \eqref{eq:LCV} may become indefinite due to the presence of subtraction. That is, positive definiteness of \eqref{eq:LCV} in the finite-sample regime cannot be guaranteed even if its constituent single-fidelity sample covariance matrices are definite. %

\section{Log-Euclidean Multi-Fidelity (LEMF) Covariance Estimator}
\subsection{Log-Euclidean Geometry and its Benefits}\label{sec:LEMF:Geometry}
We begin by orienting ourselves within the log-Euclidean geometry for $\mathbb{S}^d_{++}$ as defined in \citet{ArsignyEtAl2006}. For $\mathbf{A}$, $\mathbf{B} \in \mathbb{S}^d_{++}$ and $\lambda \in \mathbb{R}$ we have notions of \textit{logarithmic} addition $\oplus$ and scalar multiplication $\odot$,
\begin{equation} 
\begin{aligned}
    \mathbf{A} \oplus \mathbf{B} = & \Exp(\Log \mathbf{A} + \Log \mathbf{B})\,, \\
    \lambda \odot \mathbf{A} = & \Exp(\lambda \cdot \Log \mathbf{A} ) = \mathbf{A}^\lambda,
\end{aligned} 
\label{eq:logops}
\end{equation} 
corresponding to regular addition and scalar multiplication of $\log \bf A$ and $\log \bf B$ in $\mathbb{S}^d$. %
Here $\Exp: \mathbb{S}^d \to \mathbb{S}^d_{++}$ is the matrix-exponential and $\Log: \mathbb{S}^d_{++} \to \mathbb{S}^d$ is its inverse \citep[Def.~2.1]{ArsignyEtAl2007}. The mappings $\Log$ and $\Exp$ place $\mathbb{S}^d$ and $\mathbb{S}^d_{++}$ in one-to-one correspondence, enabling $\mathbb{S}^d_{++}$ to make use of the vector space structure of $\mathbb{S}^d$. 
The definitions of $\oplus$ and $\odot$ \eqref{eq:logops} satisfy the axioms of a vector space, and we can equip this vector space $\mathbb{S}^d_{++}(\oplus, \odot)$ with the log-Euclidean metric, 
\begin{equation}
        d_{\mathrm{LE}}({\mathbf{A}}, {\mathbf{B}}) = \norm{ \Log {\mathbf{A}} - \Log {\mathbf{B}}}_{\mathrm{F}} \, .
        \label{eq:LogEucMetric}
\end{equation}
Computations related to the log-Euclidean geometry can be performed economically because efficient algorithms exist for computing matrix exponentials and logarithms. Efficient computation is a major advantage as compared to the geometries induced by other metrics such as the affine invariant metric \citep{Bhatia2007,pmlr-v37-huanga15}. %

\textbf{Review of other geometries for $\mathbb{S}^d_{++}$} The set of $d \times d$ symmetric positive definite matrices $\mathbb{S}^d_{++}$ forms a Riemannian manifold embedded in the vector space of $d \times d$ symmetric matrices $\mathbb{S}^d$. One can equip $\mathbb{S}^d_{++}$ with a number of geometries in addition to the Euclidean (which we have shown to be problematic for multi-fidelity estimation) and log-Euclidean geometries discussed here. There is, e.g., the affine-invariant geometry \citep{Bhatia2007}, often billed as the ``canonical choice,'' which gives rise to the affine-invariant metric,
 \begin{equation}
        d_{\mathrm{Aff}}( {\mathbf{A}},\ {\mathbf{B}} ) = \norm{\Log \left({\mathbf{A}}^{-1}{\mathbf{B}}\right)}_{\mathrm{F}} \, .
        \label{eq:AIMetric}
    \end{equation}
Other choices include the Bures-Wasserstein geometry, arising from the $L^2$-Wasserstein distance between multivariate, nondegenerate, mean-zero Gaussians \citep{MalagoEtAl2018}; and the Log-Cholesky geometry, obtained by parametrizing SPD matrices in terms of their Cholesky factors and defining a Riemannian metric on the space of lower-triangular matrices with positive diagonal entries \citep{Lin2019}. Choice of geometry for $\mathbb{S}^d_{++}$ is application-dependent and often depends on factors such as cost to compute geodesic distance, availability of the Riemannian exponential and logarithmic maps in closed form, and whether the so-called ``swelling effect'' \citep{ArsignyEtAl2006} is an issue. We choose the log-Euclidean geometry for its computational advantages, described at the end of the previous paragraph.

\subsection{Log-Euclidean Multi-Fidelity Estimation}
We construct the log-Euclidean multi-fidelity (LEMF) covariance estimator via linear control variates in the log-Euclidean geometry for $\mathbb{S}^d_{++}$, %
\begin{multline}
    \hat\bSigma_{\bn}^{\rm LEMF} = \hat \bSigma_{n_0}^{(0)} \oplus \; \bigoplus_{\ell = 1}^L \alpha_\ell \odot \left(\hat\bSigma^{(\ell)}_{n_\ell} \ominus \hat\bSigma^{(\ell)}_{n_{\ell - 1}} \right) =\\
    \Exp\left( \Log \hat{\bSigma}_{n_0}^{(0)} + \sum_{\ell=1}^L \alpha_{\ell} \left( \Log \hat{\bSigma}^{(\ell)}_{n_{\ell}} - \Log \hat{\bSigma}^{(\ell)}_{n_{\ell-1}} \right)  \right),
\label{eq:TS}
\end{multline}
where we have used the convention $\mathbf{A} \ominus \mathbf{B} = \mathbf{A} \oplus -1\odot \mathbf{B}$.
The estimator \eqref{eq:TS} has the same algebraic structure as \eqref{eq:LCV}, but is guaranteeably positive definite at finite $\bn = (n_0, n_1, \ldots, n_L)$ as the following proposition shows:
\begin{proposition} 
The log-Euclidean multi-fidelity covariance estimator \eqref{eq:TS} exists and is positive definite almost surely whenever $n_0, n_1, \dots, n_L > d$.  
\label{prop:lemfExists}
\end{proposition}
In addition to interpreting \eqref{eq:TS} as a linear control variate construction in the log-Euclidean geometry for $\bbS^d_{++}$, we can equivalently view the LEMF estimator as a Euclidean linear control variate estimate of $\Log \bSigma$ on $\bbS^d$,
\begin{equation} 
 \Log \hat\bSigma_{\bn}^{\rm LEMF}\hspace*{-0.15cm} =
     \Log \hat{\bSigma}_{n_0}^{(0)} + \sum_{\ell=1}^L \alpha_{\ell} \left( \Log \hat{\bSigma}^{(\ell)}_{n_{\ell}}\hspace*{-0.1cm} - \hspace*{-0.1cm}\Log \hat{\bSigma}^{(\ell)}_{n_{\ell-1}} \right),
 \label{eq:logTS}
\end{equation} 
which we then map to $\mathbb{S}^d_{++}$ via the matrix exponential. Variance reduction in \eqref{eq:logTS} is achieved on $\bbS^d$ via the standard Euclidean control variate framework and then propagated to $\bbS^d_{++}$ via $\Exp(\cdot)$. {Furthermore, \eqref{eq:TS} has the form of a log-Euclidean Fréchet average of $\hat{\bSigma}_{n_0}^{(0)}$ and $L$ small perturbations induced by the surrogate sample covariance matrices; see \cref{app:frechet} for details.}

The LEMF estimator requires computing $2L + 1$ matrix logarithms and one matrix exponential. %
Note that obtaining the sample sets $\{\by^{(\ell)}(\omega_i)\}_{i=1}^{n_\ell}$, $\ell = 0, \dots, L$, is typically much more computationally expensive than computing matrix logarithms and exponentials. %
\subsection{The MSE of the LEMF Estimator}
\label{sec:3.3}
For the following analysis, we assume without loss of generality that $\by^{(\ell)}, \ell = 0, \dots, L$, have zero mean, for ease of exposition. As introduced in \cref{sec:intro}, the surrogates $\by^{(1)}, \dots, \by^{(L)}$ are ordered by decreasing fidelity,
\[
1 = \rho_0 > |\rho_1| > \cdots > |\rho_L| \geq \rho_{L + 1} = 0,
\]
where $\rho_\ell$ is a multivariate generalization of the Pearson correlation coefficient,
\begin{equation}
        \rho_{\ell} = \frac{ \tr \left( \cov[\by^{(0)}(\by^{(0)})^{\top},\ (\by^{(\ell)})(\by^{(\ell)})^{\top}] \right)  }{ \sigma_0 \sigma_{\ell}}\,, %
    \label{eq:GeneralizedCorrelation}
    \end{equation}
with generalized variances defined as 
    \begin{equation}
        \sigma_{\ell}^2 = \tr(\cov[\by^{(\ell)}(\by^{(\ell)})^{\top}]), \quad \ell = 0, \dots, L.
    \label{eq:GeneralizedVariance}
    \end{equation}
The covariance of $\by^{(\ell)}(\by^{(\ell)})^\top$, $\ell = 0, \dots, L$ is a symmetric positive semidefinite linear operator from $\mathbb{S}^d$ to $\mathbb{S}^d$,
\begin{multline}
\Cov[\by^{(\ell)}(\by^{(\ell)})^\top] = \\ \E \left[ \left(\by^{(\ell)}(\by^{(\ell)})^\top - \bSigma^{(\ell)} \right) \otimes \left(\by^{(\ell)}(\by^{(\ell)})^\top - \bSigma^{(\ell)} \right) \right],
\label{eq:autocov}
\end{multline}
where we define $\bSigma^{(\ell)} = \E[\by^{(\ell)}(\by^{(\ell)})^\top]$. Similarly, the cross-covariance between $\by^{(\ell)}(\by^{(\ell)})^\top$ and $\by^{(m)}(\by^{(m)})^\top$, $\ell, m \in \{0, \dots, L\}$ is 
\begin{multline}
\Cov \left[\by^{(\ell)}(\by^{(\ell)})^\top,\, \by^{(m)}(\by^{(m)})^\top \right] = \\
\E \left[ \left(\by^{(\ell)}(\by^{(\ell)})^\top - \bSigma^{(\ell)} \right) \otimes \left(\by^{(m)}(\by^{(m)})^\top - \bSigma^{(m)} \right) \right].
\label{eq:crosscov}
\end{multline}
In writing \eqref{eq:autocov} and \eqref{eq:crosscov} we do not invoke any particular representation for linear operators on $\mathbb{S}^d$. Rather, we only assume that $\otimes$ is an outer product for symmetric matrices compatible %
with the Frobenius inner product, that is, 
\[
\tr(\mathbf{A} \otimes \mathbf{B}) = \langle \mathbf{A}, \mathbf{B} \rangle_F, \quad \mathbf{A}, \mathbf{B} \in \mathbb{S}^d.
\]

The following proposition derives the MSE of the log-Euclidean estimator in the log-Euclidean metric \eqref{eq:LogEucMetric}. 
\begin{proposition}
Assume that $\|\bSigma - \eye\|_{\mathrm{F}} < h/4$ with $h < 1$ and $\|\bSigma - \bSigma^{(\ell)}\|_{\mathrm{F}} < h/4$ for $\ell = 1, \dots, L$. Then, for coefficient vector $\balpha \in \mathbb{R}^n$ and sufficiently large numbers of samples $\bn \in \mathbb{N}^{L+1}$, the MSE in the log-Euclidean metric of the LEMF estimator defined in~\eqref{eq:TS} is
    \begin{equation}
    \begin{split}
        &\mathbb{E}\left[ d_{\mathrm{LE}}\left( \hat{\bSigma}^{\mathrm{LEMF}}_{\bn},\   \bSigma \right)^2 \right] = \E \left[ \left\|\Log \hat\bSigma_{\bn}^{\rm LEMF} - \Log \bSigma \right\|_{\rm F}^2 \right] \\
        &= \frac{\sigma_0^2}{n_0}
        +
        \sum_{\ell=1}^L  \left(\frac{1}{n_{\ell-1}} - \frac{1}{n_{\ell}}\right)(\alpha_{\ell}^2\sigma_{\ell}^2 - 2\alpha_{\ell}\rho_{\ell}\sigma_{\ell}\sigma_0) + \mathcal{O}(h^2)
        \, .
    \end{split}
    \label{eq:LogEMSETS}
    \end{equation}
\label{prop:MSETS}
\end{proposition}

Note that the technical condition $\|\bSigma - \eye\|_F < h/4$ can always be established by re-scaling $\tilde{\by}^{(\ell)} = {\mathbf{C}}^{-1/2}\by^{(\ell)}$ with a matrix $\mathbf{C} \in \mathbb{S}^d_{++}$ and using $\mathbf{C}$ as the base of the matrix logarithm and exponential in the estimator \eqref{eq:TS}; we consider $\eye$ as the base here for ease of exposition.

Below we relate the MSE of the Euclidean estimator $\hat{\bSigma}_{\bn}^{\text{EMF}}$ \eqref{eq:LCV} in the Euclidean metric with the MSE of the log-Euclidean estimator \eqref{eq:TS} $\hat{\bSigma}_{\bn}^{\text{LEMF}}$ in the log-Euclidean metric:
\begin{corollary}
    The MSE in the Euclidean metric of the Euclidean multi-fidelity estimator~\eqref{eq:LCV} with sample size vector $\bn$ and coefficients $\balpha$ is
    \begin{multline}
        \mathbb{E}\left[d_{\mathrm{F}}\left(\hat{\bSigma}^{\mathrm{EMF}}_{\bn},\bSigma\right)^2\right] = \E \left[ \left\|\hat{\bSigma}^{\mathrm{EMF}}_{\bn} - \bSigma \right\|_{\rm F}^2 \right] =  \\
        \frac{\sigma_0^2}{n_0}
        +
        \sum_{\ell=1}^L  \left(\frac{1}{n_{\ell-1}} - \frac{1}{n_{\ell}}\right)(\alpha_{\ell}^2\sigma_{\ell}^2 - 2\alpha_{\ell}\rho_{\ell}\sigma_{\ell}\sigma_0)\, .
    \label{eq:FrobMSELCV}
    \end{multline}
\label{prop:MSELCV}
\end{corollary}
\Cref{prop:MSELCV} shows that the MSE of the LEMF estimator in the log-Euclidean metric is equal to the MSE of the EMF estimator in the Euclidean metric up to first order. 
Note that it does not comment on a cross-comparison (e.g., MSE of the LEMF in the Euclidean metric and vice-versa); we will explore such comparisons numerically in Section~\ref{sec:4}.

\subsection{Optimal Sample Allocation for the LEMF}
\Cref{prop:MSETS,prop:MSELCV} show the dependence of the MSE in the log-Euclidean and Euclidean metrics of the LEMF and EMF estimators on the sample allocations $\bn$ and coefficients $\balpha$. We now derive $\bn^*$ and $\balpha^*$ that minimize the respective MSEs up to first order and so achieve an optimal sample allocation across the hierarchy of data sources. 

We start with $\balpha^*$: because the first-order approximation of the MSE~in \eqref{eq:LogEMSETS}  and the MSE in \eqref{eq:FrobMSELCV} are both quadratic in the coefficients $\alpha_1, \dots, \alpha_L$, we can directly minimize for $\alpha_1, \dots, \alpha_L$ and obtain
\begin{equation}
        \balpha^{\star} =  [\alpha_1^{\star},\ldots,\alpha_L^{\star}] \, ,\quad \alpha_{\ell}^{\star} = \rho_{\ell}\frac{\sigma_0}{\sigma_{\ell}} \, ,
    \label{eq:OptimalCoef}
    \end{equation}
which are independent of the sample sizes.
    
In preparation for finding $\bn^*$, we define the costs of the multi-fidelity estimators as 
\begin{equation}
    c(\bn) = \sum\nolimits_{\ell = 0}^{L} n_{\ell} c_{\ell} \, .
    \label{eq:TotalCost}
    \end{equation}
    Following~\citet{PWG16}, \cref{thm:OptimalAllocation} derives the optimal sample allocation $\bn^*$ to obtain an estimator with first-order minimal MSE and costs $c(\bn^*) \leq B $ for a fixed computational budget $B > 0$.
    
{Note that in writing \eqref{eq:TotalCost}, we assume that the costs of generating \emph{samples} of the random variables $\by^{(0)}, \ldots, \by^{(L)}$ are much higher than the cost of constructing a multi-fidelity covariance estimate, and in particular eclipse the cost of computing the matrix exponential and the $2L + 1$ matrix logarithms required by the LEMF estimator \eqref{eq:TS}. This situation is often the case, e.g., in science and engineering applications when generating samples may correspond to evaluating computationally intensive physics-based models. One could explicitly account for the cost of constructing the LEMF estimator by inserting a modified budget of $\hat B = B - c_{\rm log/exp}$ into \cref{thm:OptimalAllocation}, where $c_{\rm log/exp}$ is the cost of the matrix logarithms and exponential, and subsequently obtain a result analogous to \cref{cor:MSEComparison} with a more complicated condition for when the LEMF estimator outperforms an equivalent-cost high-fidelity estimator. }

\begin{proposition}
Let $B > 0$ denote the computational budget and let \Cref{prop:MSETS} apply.  The first-order optimal sample allocation $\bn^{\star}$ for the LEMF estimator~\eqref{eq:LCV} that solves
\begin{equation}
\begin{split}
    \underset{\bn \in \mathbb{R}^{L + 1}}{\min}&\quad \mathbb{E}\left[ d_{\mathrm{LE}}\left(\hat{\bSigma}^{\mathrm{LEMF}}_{\bn},\bSigma\right)^2 \right] \\
    \text{ such that }&\quad c(\bfn) \le B \, 
\end{split}
\label{eq:OptimalAllocationProblem}
\end{equation}
is
\begin{equation}
    n_{\ell}^{\star} = {B \sqrt{ \frac{c_{0}(\rho_{\ell}^2 - \rho_{\ell+1}^2)}{c_{\ell}(1 - \rho_1^2)} }}  
    \Bigg/ { \sum_{i=0}^L c_i \sqrt{ \frac{c_0(\rho_{i}^2 - \rho_{i+1}^2)}{c_i(1 - \rho_1^2)} } }\,.\label{sec:LEMF:OptiSamp}
\end{equation}
\label{thm:OptimalAllocation}
\end{proposition}
Because the MSE of the Euclidean estimator (in $d_{\mathrm F}(\cdot, \cdot)$) is equal to the first-order approximation of the MSE of the log-Euclidean estimator (in $d_{\rm LE}(\cdot, \cdot)$), we directly obtain that \eqref{sec:LEMF:OptiSamp} is the optimal sample allocation for the Euclidean estimator as well. %
In practice, we round either up or down and so use either $\lceil \bn^{\star} \rceil$ or $\lfloor \bn^{\star} \rfloor$. %

If the correlation coefficients $\rho_l$ and costs $c_{\ell}, \ell = 0, \dots, L$ are known, one can determine the optimal $\balpha^{\star}$ and $\bn^{\star}$ for fixed budget $B$; or, alternatively, determine $\balpha^{\star}$ and $\bn^{\star}$ so that the MSE of the LEMF estimator is below a threshold $\epsilon > 0$. In situations when $\rho_\ell$ and $c_\ell$ are unknown, it is common %
to estimate them in pilot studies \citep{cliffe_multilevel_2011,PWG16} and even to reuse the pilot samples in the actual estimator \citep{KONRAD2022110898}.

\subsection{Discussion of the LEMF Estimator}
\textbf{Interpretation} \cref{thm:OptimalAllocation} shows that when the correlation coefficients~\eqref{eq:GeneralizedCorrelation} are close to 1, then the LEMF estimator~\eqref{eq:TS} allocates more of the budget to the cheaper surrogate samples. In the extreme case where $\rho_{\ell} = 0$, the entire budget is allocated to obtaining high-fidelity samples and one recovers the single-fidelity estimator~\eqref{eq:SampleCov} with $n = B/c_0$. 

The following corollary shows that the corrrelation coefficients $\rho_\ell$ and costs $c_\ell$ dictate whether the proposed LEMF estimator leads to lower MSE than the equivalent-cost single (high)-fidelity estimator. 
\begin{corollary}
Let \Cref{thm:OptimalAllocation} apply. Then, the first-order term of the MSE of the LEMF estimator with $\bn^*$ is
\begin{equation}
    \mathbb{E}\left[ d_{\mathrm{LE}}\left(\hat{\bSigma}^{\mathrm{LEMF}}_{\bn^{\star}},\bSigma\right)^2 \right] 
    \dot{=}
    \frac{\sigma_0^2}{ B} \left( \sum\limits_{\ell=0}^L  \sqrt{ c_{\ell}(\rho_{\ell}^2 - \rho_{\ell+1}^2) }\right)^2  \, ,\label{eq:FirstOrderMSELEMF}
\end{equation}
which will be smaller than the first-order log-Euclidean MSE of the cost-equivalent single-fidelity estimator~\eqref{eq:SampleCov} that uses $n_0 = B/c_0$ high-fidelity samples if and only if
\begin{equation}
    \sum\nolimits_{\ell=0}^L \sqrt{ \frac{c_{\ell}}{c_0}(\rho_{\ell}^2 - \rho_{\ell+1}^2) } < 1\, .\label{eq:WhenBetterBound}
\end{equation}
\label{cor:MSEComparison}
\end{corollary}

\cref{eq:WhenBetterBound} demonstrates that for fixed $c_0$, the benefit of using the LEMF estimator \eqref{eq:TS} over the equivalent-cost high-fidelity estimator \eqref{eq:SampleCov} is a function of both the surrogate model correlations with the high-fidelity model, $\rho_1, \dots, \rho_L$, and the surrogate model costs, $c_1, \dots, c_L$. It is particularly instructive to consider the bifidelity case with only one surrogate model ($L = 1$). In this setting, condition \eqref{eq:WhenBetterBound} simplifies to 
\begin{equation}
\sqrt{1 - \rho_1^2} + \sqrt{\frac{c_1}{c_0}\rho_1^2} < 1 \iff 2\sqrt{\frac{1 - \rho_1^2}{\rho_1^2}} < \frac{c_0 - c_1}{\sqrt{c_0c_1}}.
\label{eq:WhenBetterBound_bifidelity}
\end{equation}
If the normalized reduction in cost $\frac{c_0 - c_1}{\sqrt{c_0c_1}}$ achieved by the low-fidelity model is large, then the minimum required generalized correlation $\rho_1$ such that condition \eqref{eq:WhenBetterBound_bifidelity} is satisfied decreases. Conversely, in the limit $\rho_1 \to 1$ we only require $c_1 < c_0$ in order to see a benefit from using the LEMF estimator. In \cref{fig:cost_condition} we visually demonstrate this relationship by plotting the region in $(\rho, c_1/c_0)$ space for which \eqref{eq:WhenBetterBound_bifidelity} holds and the LEMF estimator has a lower first-order MSE than the equivalent-cost high-fidelity estimator.

\begin{figure}[h]
    \centering
    \includegraphics[width=0.9\linewidth]{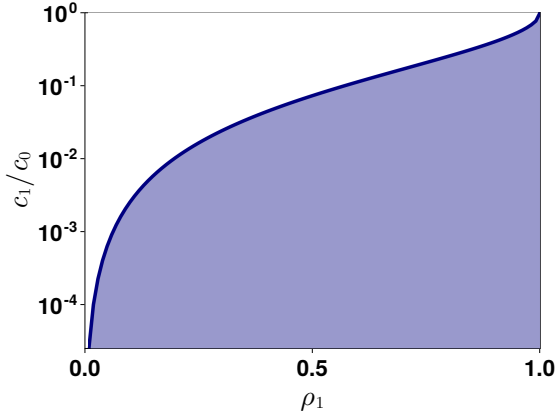}
    \caption{The shaded region corresponds to pairs $(\rho_1, c_1/c_0)$ of correlation coefficient and cost ratio for which \eqref{eq:WhenBetterBound_bifidelity} holds and thus the LEMF estimator has lower first-order error than the equivalent-cost single-fidelity estimator. Note that a lower correlation $\rho_1$ requires a smaller cost ratio $c_1/c_0$ for the LEMF estimator to achieve lower error than the single-fidelity estimator.}
    \label{fig:cost_condition}
\end{figure}
    
\textbf{Truncated multi-fidelity estimator} In addition to the EMF estimator and equivalent-cost single-fidelity estimators, we will compare the proposed LEMF estimator to the EMF estimator with positive definiteness enforced via eigenvalue truncation. %
Following \citet{HoelEtAl2016}, we define the truncated multi-fidelity covariance estimator as
    \begin{equation}
        \hat{\bSigma}^{+\delta}_{\bn} = \mathcal{T}_{\delta}\left(\hat{\bSigma}^{\mathrm{EMF}}_{\bn} \right) \, ,
    \label{eq:TruncEig}
    \end{equation}
where $\delta > 0$ is a small threshold and the truncation operator $\mathcal{T}_{\delta}$ acts on symmetric matrices via
$    \mathcal{T}_{\delta}( {\mathbf{A}} ) = {\mathbf{Q}} \left( {\boldsymbol \Lambda} \vee \delta \right){\mathbf{Q}}^{\top}$
with ${\mathbf{A}} = {\mathbf{Q}} {\boldsymbol \Lambda} {\mathbf{Q}}^{\top}$ denoting the eigen-decomposition of $\bf A$ and ${\boldsymbol \Lambda} \vee \delta$ denoting the diagonal matrix whose diagonal entries are $\max(\lambda_{ii}, \delta)$; see the Introduction for a discussion about drawbacks of the truncated estimator. %

\begin{table*}[!ht]
\centering
\begin{tabular}{|c|r|r|r|r|r|}
\hline
MSE & high-fidelity only & surrogate only & Euclidean MF & truncated MF ($\delta = 10^{-16}$) & LEMF (ours) \\
\hline
log-Euclidean & 1.72 & 5.82 &NaN & 66.88 & 0.83 \\
affine-invariant & 4.99 & 5.82 &NaN & 125.11 & 2.69 \\
Euclidean & 3.00 & 4.00 & 0.51 & 0.51 & 0.71 \\
\hline
\end{tabular}
\caption{Motivating example: The LEMF estimator leverages cheap surrogate samples to reduce the MSE compared to cost-equivalent single-fidelity estimators and other multi-fidelity estimators that rely on post-processing via truncation. } %
\label{table:ToyGaussianMSE}
\end{table*}

\begin{figure*}[!ht]
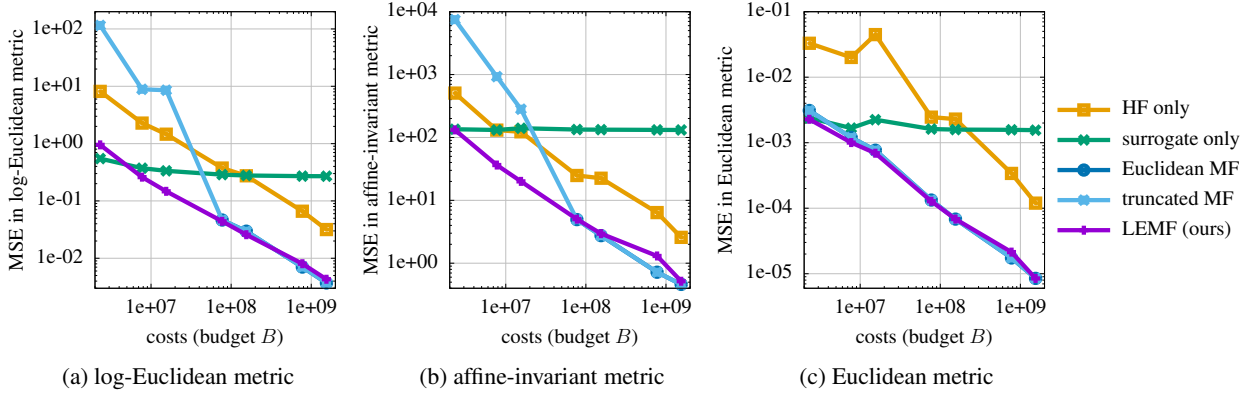

\begin{center}
\setlength\tabcolsep{1pt}
\begin{tabular}{cccc}
\resizebox{0.3\textwidth}{!}{\Large \input{figures/fwduq_heat_logE_mse.tex}}
&
\hspace*{-0.5cm}\resizebox{0.3\textwidth}{!}{\Large \input{figures/fwduq_heat_aff_mse.tex}}
&
\hspace*{-0.5cm}\resizebox{0.3\textwidth}{!}{\Large \input{figures/fwduq_heat_frob_mse.tex}}
&
\hspace*{-1.1cm}\resizebox{0.4\textwidth}{!}{\Large \input{figures/fwd_uq_mse_key.tex}}\\
 \small (a) log-Euclidean metric &  \hspace*{-0.25cm}\small (b) affine-invariant metric &  \hspace*{-0.25cm}\small (c) Euclidean metric  &
\end{tabular}
\end{center}
\caption{Heat flow: In the log-Euclidean metric, our LEMF estimator achieves the MSE tolerance below $\approx 10^{-1}$ with $\approx 30\times$ speedup compared to the single-fidelity estimator that uses high-fidelity samples alone. Only using surrogate samples leads to a large bias that prevents reaching an MSE below $10^{-1}$. The Euclidean multi-fidelity estimator is indefinite in more than 10\% of the 100 trials used here and therefore does not provide a valid covariance matrix in this example. (Plot with min/max over 100 trials is shown in Appendix~\ref{appx:UQ}.)}
\label{fig:fwd_uq_mse}
\end{figure*}
\begin{figure*}[!t]
\begin{tabular}{ccc}
\hspace*{-0.4cm}\resizebox{0.35\textwidth}{!}{\LARGE \input{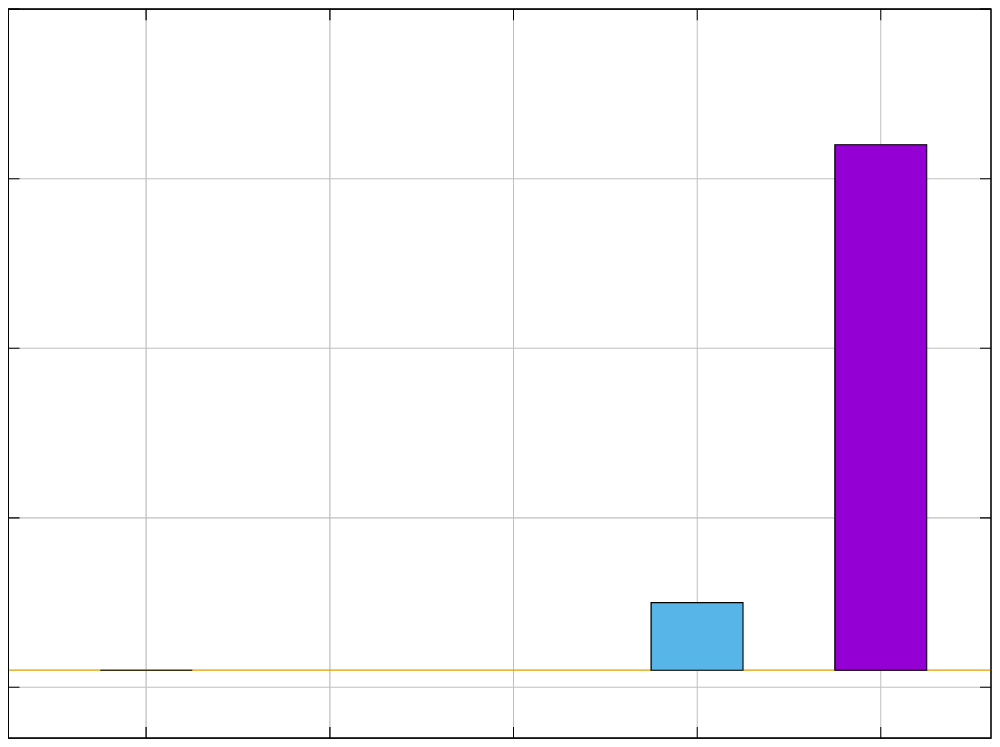}}
&
\hspace*{-0.75cm}\resizebox{0.35\textwidth}{!}{\LARGE 
\input{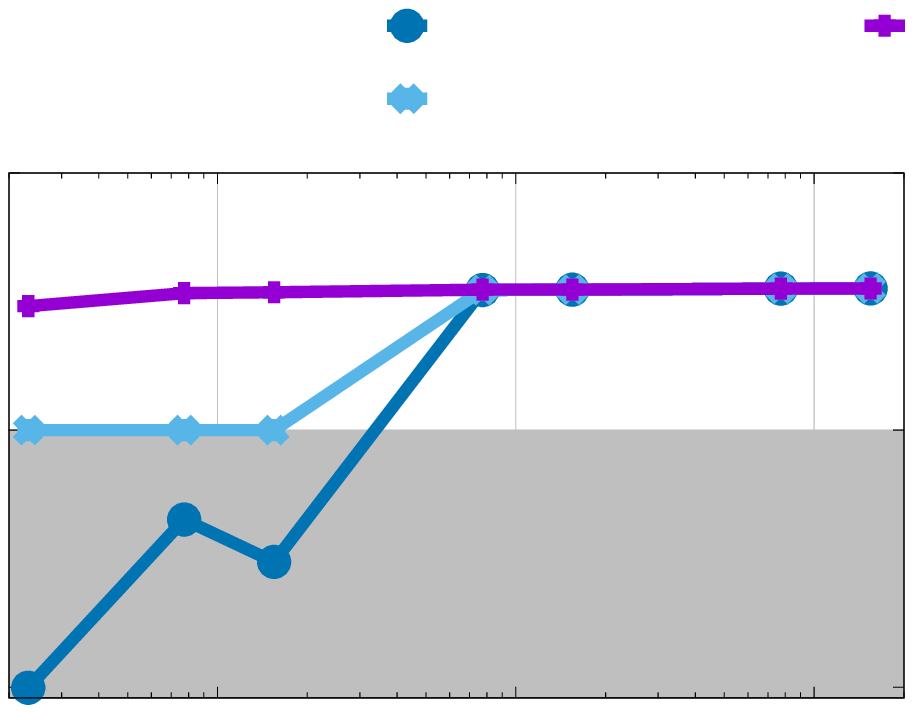}}
&
\hspace*{-0.75cm}\resizebox{0.35\textwidth}{!}{\LARGE
\input{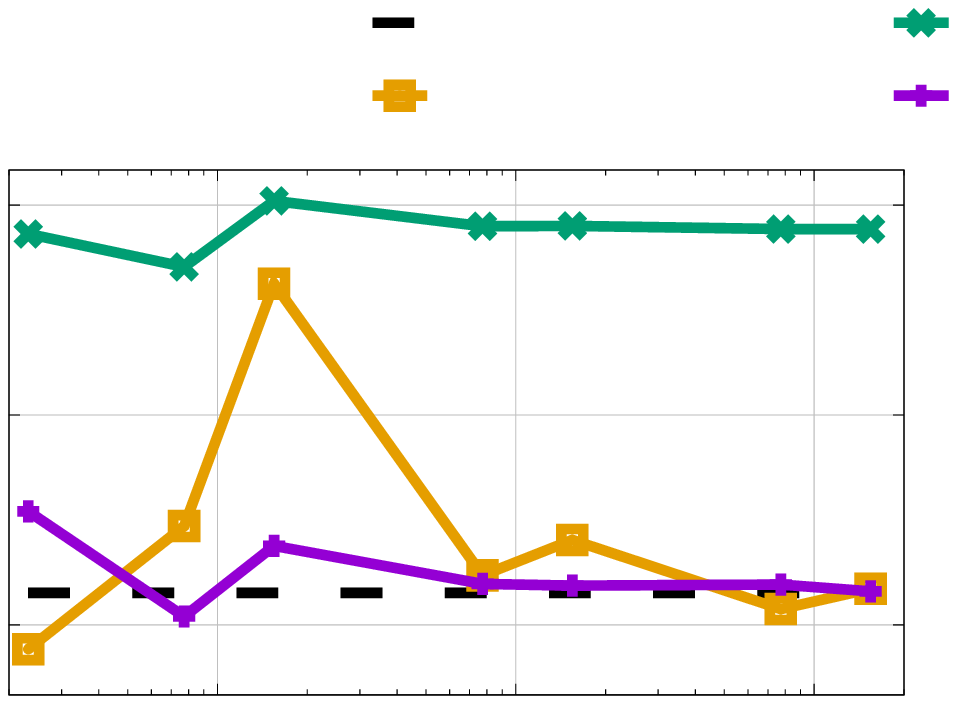}}\\
\small (a) speedup to reach MSE below $\approx 10^{-1}$ & \hspace*{-0.4cm}\small (b) smallest eigenvalue & \hspace*{-0.4cm}\small (c) convergence of  estimators
\end{tabular}
\caption{Heat flow: Plot (a): Our LEMF estimator achieves a $32\times$ speedup. Plot (b): Our LEMF estimator remains positive definite even for small sample sizes, whereas the Euclidean estimator becomes indefinite. Plot (c): Entries of our LEMF estimates converge quickly to the entries of the true covariance, whereas surrogate-only estimators incur a large bias and high-fidelity-only estimators incur high costs.} %
\label{fig:fwd_uq_other}
\end{figure*}

\section{Numerical Examples}
\label{sec:4}
\subsection{Motivating Example with Gaussians}\label{sec:NumExp:Gaussian}
Consider the problem of estimating the covariance of a Gaussian random variable $\by \sim N(\zeros, \bSigma)$, with $\bSigma \in \mathbb{S}^4_{++}$. For this motivating example, we obtain low-fidelity samples $\by^{(\ell)} = \by + \beps^{(\ell)}$ by corrupting a high-fidelity sample $\by$ with independent noise at different noise levels $\beps^{(\ell)} \sim N(\zeros, \bGamma^{(\ell)})$ for $\ell=1,2,3$. The matrices $\bGamma^{(1)}, \dots, \bGamma^{(3)}$ are chosen such that noise levels lead to generalized correlations of $\rho_1 \approx 0.93$, $\rho_2 \approx 0.74$, and $\rho_3 \approx 0.58$; see Appendix~\ref{appx:ToyGaussian}. We impose unit costs for drawing a high-fidelity sample $\by$ and costs $c_1 = 10^{-2},c_2 = 10^{-3}$, and $c_3 = 10^{-4}$ for drawing the surrogate samples. Because the samples are Gaussian, their empirical second moments follow a Wishart distribution and the generalized variances~\eqref{eq:GeneralizedVariance} and correlations~\eqref{eq:GeneralizedCorrelation} can be computed exactly; see Appendix~\ref{appx:ToyGaussian}. Table~\ref{table:ToyGaussianMSE} shows the MSE of the multi-fidelity estimators compared with that of single-fidelity covariance estimators using high-fidelity or surrogate samples alone with the same budget $B = 15$. The LEMF estimator has $2\times$ lower MSE than the other estimators in the log-Euclidean and affine-invariant metric, while being competitive in the Euclidean metric.

\subsection{Uncertainty Quantification of Steady-State Heat Flow}
\textbf{Physics model} The steady-state heat equation over the domain $\mathcal{X} = (0,1)$ with a variable heat conductivity is
    \begin{equation}
        -\operatorname{div}\left( \exp(\kappa(x; \btheta)) \nabla u(x; \btheta)  \right) = f(x) ,\quad x \in \mathcal{X}\,,
    \label{eq:heat_pde}
    \end{equation}
    where we impose Dirichlet boundary conditions $u(0,\btheta)=0,u(1,\btheta) = 1$ and constant source $f(x) = 1$. We model the log heat conductivity $\kappa$ as a degree-$4$ sine-polynomial $\kappa(x; \theta) = \sum_{k=1}^4 \theta_k \sin(2\pi k x)$, whose coefficients are given by the components of $\btheta = [\theta_1,\ldots,\theta_4]^{\top} \in \mathbb{R}^4$. %
    
    We obtain a high-fidelity observation model $G^{(0)}: \R^4 \to \R^{10}$ by approximating the solution $u(\cdot, \btheta)$ via a second-order finite difference scheme with 65,536 grid points and observing the resulting temperature at $d=10$ equally spaced points $x_i = i/11$, $i\in \{1,\ldots,10\}$, in the interior of the domain $\mathcal{X}$ (see Figure~\ref{fig:HeatSol} in Appendix~\ref{appx:UQ}). The surrogate model $G^{(1)}$ is obtained by approximating the solution $u(\cdot, \btheta)$ with only 1,024 grid points and performing the same measurement process. We identify the costs of the high-fidelity and surrogate models %
    with the number of grid points involved, $c_0 = 2^{16}$ and $c_1 = 2^{10}$. %
    
    \textbf{Covariance estimation} We model $\btheta \sim N(\zeros, \eye_{4\times 4})$ and define high-fidelity samples as $\by^{(0)} = G^{(0)}(\btheta)$ and surrogate samples as $\by^{(1)} = G^{(1)}(\btheta)$. Our aim is to estimate the covariance of $\by^{(0)}$, which is a $10 \times 10$ matrix. We first estimate the generalized variances and correlations $\sigma_0,\sigma_1,\rho_1$ in a pilot study %
    to compute the optimal multi-fidelity sample allocations $\bn^{\star} = [n_{\star}^{(0)}, n_{\star}^{(1)}]$ and $\balpha^{\star}$ for seven values of budget $B$ in the interval $[10^6, 2 \times 10^9]$; see Appendix~\ref{appx:UQ}.  We then obtain $n_{\star}^{(1)} \geq n_{\star}^{(0)}$ realizations $\btheta_1, \dots, \btheta_{n_{\star}^{(1)}}$ of the input random variable $\btheta$ and compute the $n_{\star}^{(0)}$ high-fidelity samples $\{\by^{(0)}(\btheta_i)\}_{i = 1}^{n_{\star}^{(0)}}$ and $n_{\star}^{(1)}$ surrogate samples $\{\by^{(1)}(\btheta_i)\}_{i = 1}^{n_{\star}^{(1)}}$. Note that the first $n_{\star}^{(0)}$ surrogate samples use the same realizations of $\btheta$ as the high-fidelity samples, creating statistical coupling, and thus correlation, between the high-fidelity and surrogate samples.
    
    \textbf{Results} Figure~\ref{fig:fwd_uq_mse} compares the MSE in the log-Euclidean, affine-invariant, and Euclidean metrics \eqref{eq:AIMetric} of the single-fidelity and multi-fidelity estimators over 100 trials. Using the surrogate samples alone leads to estimates with large bias with respect to the true (high-fidelity) covariance. For example, the bias of the surrogate-only estimator can be seen in Figure~\ref{fig:fwd_uq_other}c, which displays the convergence of each estimate of $\bSigma_{2,5}$ to the truth; see also Appendix~\ref{appx:UQ}. 
    
    The EMF estimator is unbiased and achieves competitively low MSE in the Euclidean metric but becomes indefinite in more than 10\% of trials (see Appendix~\ref{appx:UQ}). For this reason its MSE in the log-Euclidean and affine-invariant metrics is infinite/nonexistent (see Figure~\ref{fig:fwd_uq_other}b), as the log-Euclidean and affine-invariant distances between an element of $\bbS^d_{++}$ and any non positive-definite matrix are $\infty$. When we ``fix'' the EMF estimator via eigenvalue truncation (using a threshold of $\delta = 10^{-16}$ in \eqref{eq:TruncEig}), obtaining the positive-definite truncated EMF estimator, we introduce the possibility of large MSE in the log-Euclidean and affine-invariant metrics because we have modify the spectrum, which can be seen well for small budgets in Figure~\ref{fig:fwd_uq_mse}.
    
    By contrast, the proposed log-Euclidean multi-fidelity estimator uses surrogate samples in conjunction with high-fidelity samples and inherently retains positive definiteness, thus achieving about one order of magnitude speedup across metrics as compared to the estimator that uses the high-fidelity samples alone; Figure~\ref{fig:fwd_uq_other}a. The MSE of the LEMF estimator decays as $\mathcal{O}(1/B)$, in consort with the other Monte Carlo estimators which use high-fidelity samples. 

\begin{figure}[t]
\centering\includegraphics[width=0.90\columnwidth]{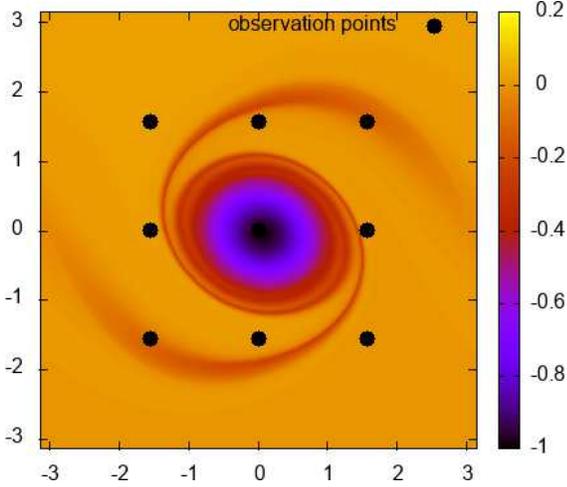}
\caption{Metric learning: An example of the final-time buoyancy for class $i = 1$. Our measurements consist of buoyancy values at nine spatial locations in the domain, indicated by black dots in the plot above. We use the observations to estimate a metric which distinguishes between observations corresponding to $\boldsymbol \theta$ sampled from class $i = 0$ and from class $i = 1$.}
\label{fig:SQG:Solution}
\end{figure}
    
\begin{figure}[t]
\centering
\resizebox{0.95\columnwidth}{!}{\large\input{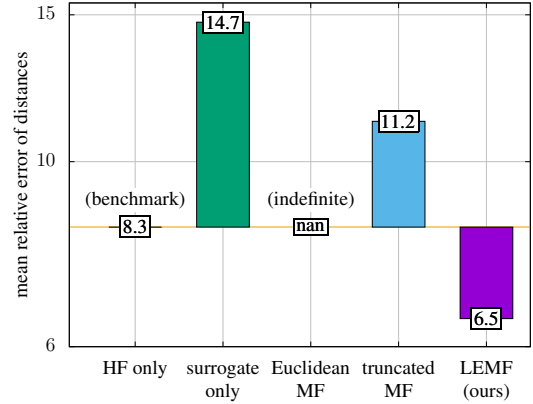}}
\caption{Metric learning: When using the learned metric to compute distances between observations of the surface quasi-geostrophic model, our LEMF estimator achieves a $>20\%$ lower error in the distances than using high-fidelity samples alone. The Euclidean estimator provides an indefinite metric in this example.}
\label{fig:SQG_metric_relerr_dist}
\end{figure}
    
\begin{figure*}[!t]
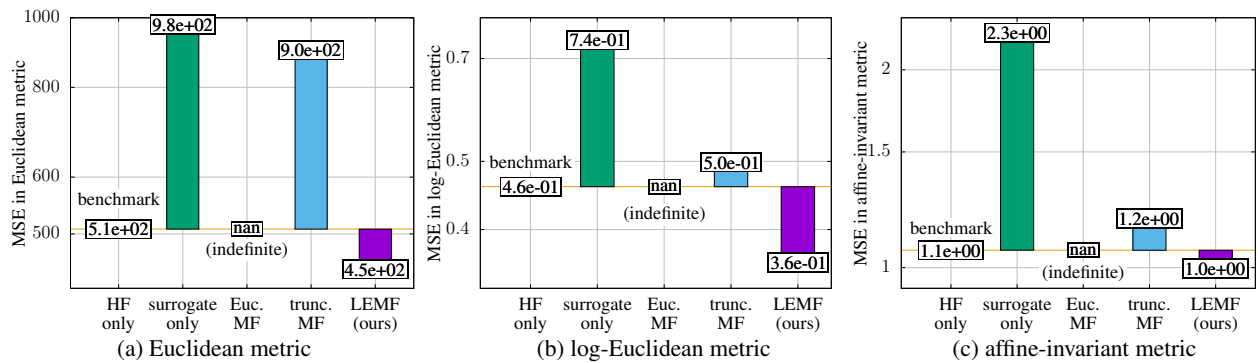

\centering
\begin{tabular}{ccc}
\resizebox{0.36\textwidth}{!}{\Large \input{figures/sqg_metric_mse_frob.tex}}
&
\hspace*{-1cm}\resizebox{0.36\textwidth}{!}{\Large \input{figures/sqg_metric_mse_logE.tex}}
&
\hspace*{-1cm}\resizebox{0.36\textwidth}{!}{\Large \input{figures/sqg_metric_mse_aff.tex}}\\
\small (a) Euclidean metric &  \hspace*{-0.5cm}\small (b) log-Euclidean metric &  \hspace*{-0.5cm}\small (c) affine-invariant metric
\end{tabular}
\caption{Metric learning: Our LEMF estimator achieves the lowest MSE compared to the cost-equivalent benchmark estimator that uses high-fidelity samples alone. The Euclidean multi-fidelity estimator is indefinite over 100 trials and therefore fails to learn a valid metric.}
\label{fig:sqg_metric}
\end{figure*}

\subsection{Metric Learning: Surface Quasi-Geostrophy}
We now apply our multifidelity covariance estimators in an intermediate step in the metric learning of \citet{Zadeh2016} to enable easier classification of fluid regimes from %
buoyancy measurements given by a surface quasi-geostrophic model \citep{HeldEtAl1985}.

\textbf{Physics model}
In the model presented in \citet{HeldEtAl1985}, the evolution of the buoyancy $b(\bx,t)$ over a periodic spatial domain $\mathcal{X} = [-\pi,\pi]\times[-\pi,\pi]$ is governed by %
    \begin{equation}
        \partial_t b(\bx, t; \btheta) + J(\psi, b) = 0 \, ,
    \label{eq:SQG}
    \end{equation}
    where $\psi$ is the streamfunction and $\bx=(x,y)$. The initial buoyancy is
    $
        b_0(\bx;\btheta) = -\frac{1}{(2\pi/ |\theta_5|)^2} \exp\left( -x^2 - \exp(2\theta_1) y^2 \right),
    $
the contours of which form ellipses parametrized by $\btheta = [\theta_1,\ldots,\theta_5]^{\top} \in \mathbb{R}^5$. 
$\theta_1$ is the log-aspect ratio and controls the origination of vortices in the buoyancy over time, and $\theta_5$ controls the amplitude of $b_0$. The parameters $\theta_2$, $\theta_3$, and $\theta_4$ determine additional aspects of the dynamics \eqref{eq:SQG}; see \cref{appx:SQG} for details. We sample the parameters $\btheta \in \mathbb{R}^5$ from a two-component Gaussian mixture whose means differ only in the coordinate corresponding to the log-aspect ratio; see Appendix~\ref{appx:SQG}. We use $i \in \{0,1\}$ to denote the mixture component from which a given realization of $\btheta$ is sampled. 

The high-fidelity parameter-to-observable map $G^{(0)}$ maps $\btheta$ onto $d = 9$ observations of the numerical solution of  equation~\eqref{eq:SQG} in the spatial domain; see Figure~\ref{fig:SQG:Solution} as well as Figure~\ref{fig:SQG_solution} in Appendix~\ref{appx:SQG}. For $G^{(0)}$, numerical solution of \eqref{eq:SQG} is computed up to time $T=24$ using finite differences with $256$ grid points along each coordinate dimension and with a time-step of $\Delta t = 0.005$ \footnote{We use the Python package \texttt{pyqg} to compute our solutions; see http://pyqg.readthedocs.org/}. High-fidelity samples $\by^{(0)}$ are thus given by $\by^{(0)} = G^{(0)}(\btheta)$ and have covariance $\bSigma \in \bbS^9_{++}$. 

The surrogate parameter-to-observable map $G^{(1)}$ and random variable $\by^{(1)}$ are defined in the same way except that in solving \cref{eq:SQG} we only use $64$ grid points along each coordinate dimension. Evaluating $G^{(1)}$ is $16$ times faster than evaluating $G^{(0)}$, with $c_0 = 256^2$ and $c_1 = 64^2$.

\textbf{Metric learning}
We now apply our covariance estimator~\eqref{eq:TS} to learn a metric on the data manifold. %
We use a slight modification of the geometric mean metric learning approach introduced in~\citet{Zadeh2016}, which learns an SPD matrix ${\mathbf{A}}$ to define the metric
    \[
        d_{{\mathbf{A}}}(\by_1, \by_2) = \sqrt{ (\by_1 - \by_2)^{\top} {\mathbf{A}} (\by_1 - \by_2) } \, . 
    \]
The matrix ${\mathbf{A}}$ is an interpolation in the affine-invariant geometry for $\bbS^d_{++}$ between the inverse of a similarity matrix
    $
        {\mathbf{S}} = \cov[\by \mid i = 0] + \cov[\by \mid i = 1] 
    $
    and a dissimilarity matrix $
        {\mathbf{D}} = {\mathbf{S}} + {\boldsymbol \mu} {\boldsymbol \mu}^{\top}
    $,
    where
    $
        {\boldsymbol \mu} = \mathbb{E}[\by \mid i = 0] - \mathbb{E}[\by \mid i = 1] \, .
    $
Following \citet{Zadeh2016}, this interpolation is given by
    \[
        {\mathbf{A}} = {\mathbf{S}}^{-1/2}({\mathbf{S}}^{1/2} {\mathbf{D}} {\mathbf{S}}^{1/2})^t {\mathbf{S}}^{-1/2}\, ,\quad t \in [0,1] \, .
    \]
We apply our multifidelity techniques to estimating the two covariance matrices that comprise $\bf S$, $\Cov[\by \mid i = 0]$ and $\Cov[\by \mid i = 1]$. In constructing $\bf A$
we set $t = 0.1$, motivated by \citep[Figure~3]{Zadeh2016}.

\textbf{Results}
Prior to applying our multi-fidelity estimators in this setting, we take a total of $24,000$ pilot samples between $(\by^{(0)}, \by^{(1)}) \mid i = 0$ and $(\by^{(0)}, \by^{(1)}) \mid i = 1$ in order to estimate the generalized variances $\sigma_0^2$ and $\sigma_1^2$ and correlation $\rho_1$ in both classes. We additionally use these pilot samples to obtain a reference estimate of $\bf A$, which we denote $\mathbf{A}_0$. %

We estimate $\Cov[\by^{(0)} \,|\, i = 0]$ and $\Cov[\by^{(0)} \,|\, i = 1]$ as follows: 
{we construct multifidelity estimators from a combination of 15 high-fidelity samples and an additional number $n_{\mathrm{lo}, i}$ of low-fidelity samples computed according to the optimal ratio in \cref{thm:OptimalAllocation}. We correspondingly construct equivalent-cost single-fidelity estimators with budgets
\[
B_i = 15c_{\rm hi} + n_{\mathrm{lo}, i}c_{\rm lo}, \quad i \in \{0,1\},
\]
spent entirely on high- or low-fidelity samples. We estimate $\Cov[\by \mid i = 0]$ and $\Cov[\by \mid i = 1]$ using the EMF, LEMF, and truncated EMF ($\delta = 10^{-16}$) estimators and the high-fidelity- and low-fidelity-only estimators with the sample allocations described above.
}
The two covariance estimates are combined to obtain estimates of the metric matrix $\bf A$.  

We first study the performance of the learned metrics relative to the reference metric. Let $\{\by^{(0)}_j\}_{j=1}^{5000} \in \mathbb{R}^9$ be a test set of 5,000 points and consider the mean relative error of distances $\mathrm{MRE}({\mathbf{A}})$ with respect to the reference metric ${\mathbf{A}}_0$; see Appendix~\ref{appx:SQG} for a definition. 
Figure~\ref{fig:SQG_metric_relerr_dist} shows the MRE for each estimator averaged over 50 independent trials. The LEMF estimator outperforms the single-fidelity estimator and the truncated MF estimator~\eqref{eq:TruncEig} by more than 20\%. The Euclidean multi-fidelity estimator gives indefinite estimates of $\bf A$, which cannot constitute valid metrics. %

We further compare the MSEs of the learned metric matrices %
relative to the reference ${\mathbf{A}}_0$. Figure~\ref{fig:sqg_metric} shows that the LEMF estimator~\eqref{eq:TS} yields a closer approximation of the reference metric; in particular the LEMF estimator gives a more accurate estimate of $\bf A$ relative to the high-fidelity-only estimator while the other estimators give less accurate estimates. The surrogate-only estimator is limited by its large bias. The EMF estimator~\eqref{eq:LCV} is again indefinite in $>$10\% of trials. %

\section{Conclusions}
We showed that formulating multi-fidelity estimation in a non-Euclidean geometry can be beneficial for enforcing structure: The proposed LEMF estimator leverages the log-Euclidean geometry for SPD matrices to guarantee that the resulting covariance matrix estimates are positive definite. This property is in contrast to state-of-the-art multi-fidelity estimators that can lose definiteness, especially in the small-sample regime, which is typical in science and engineering applications in which drawing samples is expensive. %
Our LEMF estimator preserves definiteness and leverages surrogates to achieve speedups of more than one order of magnitude in our experiments. The experiments further show that preserving definiteness is key for enabling tasks such as metric learning and classification, for which indefinite Euclidean multi-fidelity estimators are invalid. Future work includes combining the multi-fidelity approach with, e.g., shrinkage and other regularization schemes for high-dimensional covariance estimation. 

\section*{Acknowledgements}
AM and YM acknowledge support from the Office of Naval Research, SIMDA (Sea Ice Modeling and Data Assimilation) MURI, award number N00014-20-1-2595 (Dr.~Reza Malek-Madani and Dr.~Scott Harper). AM additionally acknowledges support from the NSF Graduate Research Fellowship under Grant No.\ 1745302. TA and BP acknowledge support from the AFOSR under Award Number FA9550-21-1-0222 (Dr.~Fariba Fahroo).

\section*{Code Availability} Python implementations of the LEMF, EMF, and truncated EMF estimators are available at \href{https://github.com/amaurais/LEMF}{https://github.com/amaurais/LEMF}.

\bibliographystyle{icml/icml2023.bst}
\bibliography{references}

\clearpage
\appendix
\section{Connection to Fréchet Averaging}
\label{app:frechet}
In addition to being a control-variate estimator in the log-Euclidean geometry, \cref{eq:TS} can also be viewed as a Fréchet average. Recall that the LEMF estimator takes the form
\begin{multline*}
\Log \hat\bSigma_{\bf n}^{\rm LEMF} =
     \Log \hat{\bSigma}_{n_0}^{(0)}  \\ + \sum_{\ell=1}^L \alpha_{\ell} \left( \Log \hat{\bSigma}^{(\ell)}_{n_{\ell}} - \Log \hat{\bSigma}^{(\ell)}_{n_{\ell-1}} \right).
\end{multline*}
Defining the SPD ``difference'' matrices 
\begin{multline*}
\mathbf{D}_\ell = \Exp(\Log \hat{\bSigma}^{(\ell)}_{n_{\ell}} - \Log \hat{\bSigma}^{(\ell)}_{n_{\ell-1}}) =  \hat{\bSigma}^{(\ell)}_{n_{\ell}} \ominus  \hat{\bSigma}^{(\ell)}_{n_{\ell-1}}, \\ \ell = 1, \dots, L
\end{multline*}
we equivalently have 
\begin{equation} 
\Log \hat\bSigma_{\bf n}^{\rm LEMF} =
     \Log \hat{\bSigma}_{n_0}^{(0)} + \sum_{\ell=1}^L \alpha_{\ell} \Log \mathbf{D}_\ell. 
     \label{eq:le_barycenter}
\end{equation}
Assuming that $\alpha_1, \dots, \alpha_L > 0$ \footnote{If $\alpha_\ell < 0$ for some $\ell$, we can define $\beta_\ell = -\alpha_\ell$ and redefine $\mathbf{D}_\ell = \hat{\bSigma}^{(\ell)}_{n_{\ell-1}} \ominus \hat{\bSigma}^{(\ell)}_{n_{\ell}}$ by reversing the order of subtraction in order to ensure that the Fréchet weights are all positive.}, we see in \eqref{eq:le_barycenter} by definition of the Fréchet mean in the log-Euclidean geometry \cite{ArsignyEtAl2006} that $\hat\bSigma^{\rm LEMF}_{\bf n}$ satisfies 
\begin{align*}
\hat\bSigma^{\rm LEMF}_{\bf n} &\equiv \E_{\mathrm{LE},\, {\balpha} }(\hat{\bSigma}_{n_0}^{(0)}, \mathbf{D}_1, \dots, \mathbf{D}_L) \\
&= \argmin_{\bSigma \in \bbS^d_{++}} \left( d^2_{\rm LE}(\bSigma, \hat{\bSigma}_{n_0}^{(0)}) + \sum_{\ell=1}^L \alpha_{\ell}\, d^2_{\rm LE}(\bSigma, \mathbf{D}_\ell) \right) %
\end{align*} 
where $\balpha  = \begin{bmatrix}1 & \alpha_1 & \cdots & \alpha_L\end{bmatrix}^\top \in \mathbb{R}^{L + 1}$. Thus $\hat\bSigma^{\rm LEMF}_{\bf n}$ can be interpreted as a \textit{weighted Fréchet average} between high-fidelity $\hat{\bSigma}_{n_0}^{(0)}$ and the low-fidelity perturbations $\mathbf{D}_\ell$, $\ell = 1, \dots, L$, with weights assigned according to how much $\hat{\bSigma}_{n_0}^{(0)}$ is correlated in a generalized sense with each $\mathbf{D}_\ell$. 

\section{Proof of \cref{prop:lemfExists}}
We have defined the log-Euclidean multi-fidelity estimator by 
\begin{multline}
    \hat\bSigma_n^{\rm LEMF} = \hat \bSigma_{n_0}^{(0)} \oplus \; \bigoplus_{\ell = 1}^L \alpha_\ell \odot \left(\hat\bSigma^{(\ell)}_{n_\ell} \ominus \hat\bSigma^{(\ell)}_{n_{\ell - 1}} \right) =\\
    \Exp\left( \Log \hat{\bSigma}_{n_0}^{(0)} + \sum_{\ell=1}^L \alpha_{\ell} \left( \Log \hat{\bSigma}^{(\ell)}_{n_{\ell}} - \Log \hat{\bSigma}^{(\ell)}_{n_{\ell-1}} \right)  \right).
\label{eq:LEMF_app}
\end{multline}
Because we have assumed $n_0, n_1, \dots, n_L > d$, the sample covariance matrices $\hat\bSigma^{(\ell)}_{n_\ell}$, and $\hat\bSigma^{(\ell)}_{n_{\ell - 1}}$, $\ell = 0, \dots, L$ are positive definite almost surely. That is, they are members of $\bbS^d_{++}$. Because $\bbS^d_{++}(\oplus, \odot)$ is a vector space, it follows immediately that the top line of \eqref{eq:LEMF_app}
\[
  \hat\bSigma_n^{\rm LEMF} = \hat \bSigma_{n_0}^{(0)} \oplus \; \bigoplus_{\ell = 1}^L \alpha_\ell \odot \left(\hat\bSigma^{(\ell)}_{n_\ell} \ominus \hat\bSigma^{(\ell)}_{n_{\ell - 1}} \right) \in \bbS^d_{++}
\]
is positive definite. It remains to show that the bottom line of \eqref{eq:LEMF_app} is equivalent to the top line. This we accomplish by straightforward algebra, 
\begin{multline}
    \hat\bSigma_n^{\rm LEMF} = \hat \bSigma_{n_0}^{(0)} \oplus \; \bigoplus_{\ell = 1}^L \alpha_\ell \odot \left(\hat\bSigma^{(\ell)}_{n_\ell} \ominus \hat\bSigma^{(\ell)}_{n_{\ell - 1}} \right)  \\
    = \hat \bSigma_{n_0}^{(0)} \oplus \; \bigoplus_{\ell = 1}^L \alpha_\ell \odot \left(\hat\bSigma^{(\ell)}_{n_\ell} \oplus \Exp(-\Log \hat\bSigma^{(\ell)}_{n_{\ell - 1}}) \right) \\
    = \hat \bSigma_{n_0}^{(0)} \oplus \; \bigoplus_{\ell = 1}^L \Exp\left( \alpha_\ell  \left(\Log \hat\bSigma ^{(\ell)}_{n_\ell} -\Log \hat\bSigma^{(\ell)}_{n_{\ell - 1}} \right) \right) \\
    = \hat \bSigma_{n_0}^{(0)} \oplus \Exp \left(\sum_{\ell = 1}^L\alpha_\ell  \left(\Log \hat\bSigma ^{(\ell)}_{n_\ell} -\Log \hat\bSigma^{(\ell)}_{n_{\ell - 1}} \right) \right) \\
    = \Exp\left( \Log \hat{\bSigma}_{n_0}^{(0)} + \sum_{\ell=1}^L \alpha_{\ell} \left( \Log \hat{\bSigma}^{(\ell)}_{n_{\ell}} - \Log \hat{\bSigma}^{(\ell)}_{n_{\ell-1}} \right)  \right)\,.
\end{multline}

\section{Proof of \cref{prop:MSELCV}}
We first note that, by properties of inner and outer products on $\bbS^d$, for two random matrices ${\mathbf{A}} ,{\mathbf{B}} \in \bbS^d$ with $\mathbb{E}[{\mathbf{A}}] = \mathbb{E}[{\mathbf{B}}]= \zeros$ we have
\begin{multline}
    \mathbb{E} \left[\inner{ {\mathbf{A}} }{ {\mathbf{B}} }_{\mathrm{F}} \right] = \E[\tr(\mathbf{A} \otimes \mathbf{B})] \\ = \tr(\E[\mathbf{A} \otimes \mathbf{B}]) = \tr(\Cov[\mathbf{A}, \mathbf{B}]),
\label{eq:TrCrossCovRelation}
\end{multline}
where we have used linearity of trace. 
In the special case where ${\mathbf{A}} = {\mathbf{B}}$ we have
\begin{equation}
    \mathbb{E}[ \norm{ {\mathbf{A}} }^2_{\mathrm{F}} ]
    = \E[\langle \mathbf{A}, \, \mathbf{A} \rangle_{\rm F}] =
    \tr\left( \cov[ {\mathbf{A}} ] \right) \, .
\label{eq:TrCovRelation}
\end{equation}

We now proceed with the remainder of the proposition.  %

Recall that the EMF estimator is defined 
\begin{equation*}
    \hat{\bSigma}_{\bn}^{\mathrm{EMF}}
    =
    \hat{\bSigma}_{n_0}^{(0)} + \sum\nolimits_{\ell=1}^L \alpha_{\ell} \left( \hat{\bSigma}_{n_{\ell}}^{(\ell)} - \hat{\bSigma}^{(\ell)}_{n_{\ell-1}} \right)  \, ,
    \end{equation*}
    where $\balpha = [\alpha_1,\ldots,\alpha_L]^{\top} \in \mathbb{R}^L$ are the control variate weights
    and 
    \begin{equation} 
    \begin{gathered}
    \hat{\bSigma}^{(\ell)}_{n_{\ell}} = \frac{1}{n_\ell}\sum\nolimits_{i=1}^{n_\ell} \by^{(\ell)}_i \left(\by^{(\ell)}_i \right)^{\top}, \\
    \hat{\bSigma}^{(\ell)}_{n_{\ell - 1}} = \frac{1}{n_{\ell - 1} }\sum\nolimits_{i=1}^{n_{\ell - 1}} \by^{(\ell)}_i\left(\by^{(\ell)}_i \right)^{\top}
    \end{gathered}
    \label{eq:scms_ndenom} 
    \end{equation} 
    are sample covariance matrices computed from $\{\by^{(\ell)}_i\}_{i=1}^{n_\ell}$, $\ell = 0, \dots, L$ and we have assumed that $\E[\by^{(\ell)}]$ is known and without loss of generality equal to $\bf0$, $\ell = 0, \dots, L$. Because the means of $\by^{(0)}, \dots, \by^{(L)}$ are known, in defining the SCMs in \cref{eq:scms_ndenom} we have normalized by $n_\ell$ and $n_{\ell - 1}$ to obtain unbiased sample covariance estimates.  Under this assumption the EMF estimator \eqref{eq:LCV} with SCMs defined as in \eqref{eq:scms_ndenom} is unbiased, 
    \begin{multline*} 
    \E \left[\hat{\bSigma}_{\bn}^{\mathrm{EMF}} \right]
    = \\ 
    \E \left[\hat{\bSigma}_{n_0}^{(0)} + \sum\nolimits_{\ell=1}^L \alpha_{\ell} \left( \hat{\bSigma}_{n_{\ell}}^{(\ell)} - \hat{\bSigma}^{(\ell)}_{n_{\ell-1}} \right) \right] = \bSigma. 
    \end{multline*} 
    
Because of the unbiasedness, the Frobenius norm MSE of $\hat\bSigma^{\rm EMF}_{\bn}$ is thus equal to the trace of the covariance of $\hat\bSigma^{\rm EMF}_{\bn},$
\begin{multline*} 
\E[||\hat\bSigma^{\rm EMF}_{\bn} - \bSigma||_F^2] = \E[||\hat\bSigma^{\rm EMF}_{\bn} - \E[\hat\bSigma^{\rm EMF}_{\bn}] ||_F^2] \\ = \tr(\Cov[\hat\bSigma^{\rm EMF}_{\bn}]).
\end{multline*} 

The covariance of $\hat\bSigma^{\rm EMF}_{\bn}$ can be decomposed
\begin{multline} 
\Cov[\hat\bSigma^{\rm EMF}_{\bn}] = \Cov[\hat{\bSigma}_{n_0}^{(0)}]  \\
+ \sum_{\ell = 1}^L \alpha_\ell^2 (\Cov[\hat{\bSigma}_{n_{\ell}}^{(\ell)}] + \Cov[\hat{\bSigma}_{n_{\ell - 1}}^{(\ell)}] )  \\
+ \sum_{\ell = 1}^L \alpha_\ell \left(\Cov[\hat{\bSigma}_{n_{\ell}}^{(\ell)}, \hat{\bSigma}_{n_0}^{(0)}] + \Cov[\hat{\bSigma}_{n_0}^{(0)}, \hat{\bSigma}_{n_{\ell}}^{(\ell)}]\right. \\ 
- \left. \Cov[\hat{\bSigma}_{n_{\ell - 1}}^{(\ell)}, \hat{\bSigma}_{n_0}^{(0)}] - \Cov[\hat{\bSigma}_{n_0}^{(0)}, \hat{\bSigma}_{n_{\ell - 1}}^{(\ell)}] \right) \\
+ \sum_{\ell = 1}^L \sum_{m = \ell + 1}^L \alpha_\ell\alpha_m \left(\Cov[\hat{\bSigma}_{n_{\ell}}^{(\ell)}, \hat{\bSigma}_{n_{m}}^{(m)}] - \Cov[\hat{\bSigma}_{n_{\ell - 1}}^{(\ell)}, \hat{\bSigma}_{n_{m}}^{(m)}]\right. \\ 
- \left.\Cov[\hat{\bSigma}_{n_{\ell}}^{(\ell)}, \hat{\bSigma}_{n_{m-1}}^{(m)}] + \Cov[\hat{\bSigma}_{n_{\ell - 1}}^{(\ell)}, \hat{\bSigma}_{n_{m - 1}}^{(m)}] \right) \\
+ \sum_{\ell = 1}^L \sum_{m = \ell + 1}^L \alpha_\ell\alpha_m \left(\Cov[\hat{\bSigma}_{n_{m}}^{(m)}, \hat{\bSigma}_{n_{\ell}}^{(\ell)}] - \Cov[\hat{\bSigma}_{n_{m}}^{(m)}, \hat{\bSigma}_{n_{\ell - 1}}^{(\ell)}]\right. \\ 
- \left.\Cov[\hat{\bSigma}_{n_{m-1}}^{(m)}, \hat{\bSigma}_{n_{\ell}}^{(\ell)}] + \Cov[\hat{\bSigma}_{n_{m - 1}}^{(m)}, \hat{\bSigma}_{n_{\ell - 1}}^{(\ell)}] \right) \\
- \sum_{\ell = 1}^L \alpha_\ell^2 \left(\Cov[\hat{\bSigma}_{n_{\ell}}^{(\ell)}, \hat{\bSigma}_{n_{\ell - 1}}^{(\ell)})] + \Cov[\hat{\bSigma}_{n_{\ell - 1}}^{(\ell)}, \hat{\bSigma}_{n_{\ell}}^{(\ell)}] \right)
\label{eq:covLCV}
\end{multline}

We will need a result analogous to Lemma 3.2 in \citet{PWG16} in order to simplify \eqref{eq:covLCV}. 
\begin{lemma}
Consider $\hat\bSigma^{(\ell)}_{m_i} = \frac{1}{m_i}\sum_{i' = 1}^{m_i} \by_{i'}^{(\ell)}(\by_{i'}^{(\ell)})^{\top}$ and $\hat\bSigma^{(t)}_{m_j} = \frac{1}{m_j}\sum_{j' = 1}^{m_j} \by_{j'}^{(t)}(\by_{j'}^{(t)})^{\top}$, where $\ell, t \in \{0, \dots, L\}$ and $m_i, m_j \in \mathbb{Z^+}$. It holds that 
\begin{multline*} 
\Cov[\hat\bSigma^{(\ell)}_{m_i}, \; \hat\bSigma^{(t)}_{m_j}] \\ 
= \frac{1}{\max\{m_i, m_j\}}\Cov[\by^{(\ell)}(\by^{(\ell)})^\top, \; \by^{(t)}(\by^{(t)})^\top] 
\end{multline*} 
\end{lemma}
\begin{proof}
Proceeding algebraically, we see 
\begin{multline*}
    \Cov[\hat\bSigma^{(\ell)}_{m_i}, \; \hat\bSigma^{(t)}_{m_j}] = \\ \Cov \left[\frac{1}{m_i}\sum_{i' = 1}^{m_i} \by_{i'}^{(\ell)}(\by_{i'}^{(\ell)})^{\top}, \; \frac{1}{m_j}\sum_{j' = 1}^{m_j} \by_{j'}^{(t)}(\by_{j'}^{(t)})^{\top} \right] \\ 
    = \frac{1}{m_i m_j}\Cov \left[\sum_{i' = 1}^{m_i} \by_{i'}^{(\ell)}(\by_{i'}^{(\ell)})^{\top}, \; \sum_{j' = 1}^{m_j} \by_{j'}^{(t)}(\by_{j'}^{(t)})^{\top} \right].
\end{multline*}
$\by_{i'}^{(\ell)}(\by_{i'}^{(\ell)})^{\top}$ and $\by_{j'}^{(t)}(\by_{j'}^{(t)})^{\top}$ are independent except in the case that $i' = j'$. Thus the terms that remain are 
\begin{multline*}
    \Cov[\hat\bSigma^{(\ell)}_{m_i}, \; \hat\bSigma^{(t)}_{m_j}] = \\
    \frac{1}{m_i m_j} \sum_{i' = 1}^{\min\{m_i, m_j\}} \Cov[\by_{i'}^{(\ell)}(\by_{i'}^{(\ell)})^{\top}, \; \by_{i'}^{(t)}(\by_{i'}^{(t)})^{\top}] \\
    = \frac{\min\{m_i, m_j\}}{m_i m_j} \Cov[\by^{(\ell)}(\by^{(\ell)})^{\top}, \; \by^{(t)}(\by^{(t)})^{\top}] \\
    = \frac{1}{\max\{m_i, m_j\}} \Cov[\by^{(\ell)}(\by^{(\ell)})^{\top}, \; \by^{(t)}(\by^{(t)})^{\top}].
\end{multline*}
\end{proof}

We return to \cref{eq:covLCV} with this result and simplify, seeing that the doubly-indexed sums are zero, 
\begin{multline*} 
\Cov[\hat\bSigma^{\rm EMF}_{\bn}] = \tfrac{1}{n_0}\Cov[\by^{(0)}(\by^{(0)})^\top]  \\
+ \sum_{\ell = 1}^L \alpha_\ell^2(\tfrac{1}{n_\ell} + \tfrac{1}{n_{\ell - 1}})\Cov[\by^{(\ell)}(\by^{(\ell)})^\top] \\
+ \sum_{\ell = 1}^L \alpha_\ell (\tfrac{1}{n_{\ell}} - \tfrac{1}{n_{\ell - 1}})\left(\Cov[\by^{(\ell)}(\by^{(\ell)})^\top, \; \by^{(0)}(\by^{(0)})^\top] \right. + \\
+ \left. \Cov[\by^{(0)}(\by^{(0)})^\top, \;\by^{(\ell)}(\by^{(\ell)})^\top] \right) \\
- \sum_{\ell = 1}^L \alpha_\ell^2 \tfrac{2}{n_\ell} \Cov[\by^{(\ell)}(\by^{(\ell)})^\top].  
\end{multline*}
We rearrange and write the covariance of $\hat\bSigma^{\rm EMF}_{\bn}$ compactly as  
\begin{multline*}
\Cov[\hat\bSigma^{\rm EMF}_{\bn}] = \tfrac{1}{n_0}\Cov[\by^{(0)}(\by^{(0)})^\top]  \\
+ \sum_{\ell = 1}^L (\tfrac{1}{n_{\ell - 1}} - \tfrac{1}{n_\ell}) \left(\alpha_\ell^2 \Cov[\by^{(\ell)}(\by^{(\ell)})^\top]\right. \\ 
- \alpha_\ell \left(\Cov[\by^{(\ell)}(\by^{(\ell)})^\top, \by^{(0)}(\by^{(0)})^\top]\right. \\ 
+ \left.\left. \Cov[\by^{(0)}(\by^{(0)})^\top, \by^{(\ell)}(\by^{(\ell)})^\top] \right) \right).
\end{multline*}
The MSE of $\hat\bSigma^{\rm EMF}_{\bn}$ is equal to the trace of its covariance,
\begin{multline*}
    \E[||\hat\bSigma^{\rm EMF}_{\bn} - \bSigma||] = \tr(\Cov[\hat\bSigma^{\rm EMF}_{\bn}]) \\ = \frac{\sigma_0^2}{n_0} + \sum_{\ell = 1}^L (\frac{1}{n_{\ell - 1}} - \frac{1}{n_\ell})(\alpha_\ell^2\sigma^2_\ell - 2\alpha_\ell\rho_{\ell}\sigma_0\sigma_\ell )
\end{multline*}
which is the result we sought to show.

\section{Proof of Proposition~\ref{prop:MSETS}}

\begin{proof}
Canceling the matrix logarithm with the matrix exponential in the LEMF estimator~\eqref{eq:TS} yields
\begin{align*}
&\norm{\Log \hat{\bSigma}_{\bn}^{\mathrm{LEMF}} -  \Log \bSigma}_{\mathrm{F}}
=\\
&
\norm{  \Log \hat{\bSigma}_{n_0}^{(0)} + \sum_{\ell=1}^L \alpha_{\ell} \left( \Log \hat{\bSigma}^{(\ell)}_{n_{\ell}} - \Log \hat{\bSigma}^{(\ell)}_{n_{\ell-1}} \right)  - \Log \bSigma  }_{\mathrm{F}} \, .
\end{align*}
Since $\norm{\bSigma - \eye}_{\mathrm{F}} < h/4 < 1$ we may define the matrix logarithm through its Taylor series
\begin{equation}
    \Log(\bSigma) = \sum_{k=1}^{\infty} \frac{(-1)^{k+1}}{k}(\bSigma - \eye)^k \, ,
\label{eq:Prop3.2Eq1}
\end{equation}
which converges absolutely.  Similarly, by the triangle inequality
\begin{equation}
    \norm{\bSigma^{(\ell)} - \eye}_{\mathrm{F}} 
    \le 
    \norm{\bSigma - \eye}_{\mathrm{F}}
    +
    \norm{\bSigma^{(\ell)} - \bSigma}_{\mathrm{F}}
    \le
    \frac{h}{2} < 1\, .
\label{eq:Prop3.2Eq2}
\end{equation}
For the sample covariance matrices $\bSigma^{(\ell)}_{n_{\ell}}$ and $\bSigma^{(\ell)}_{n_{\ell-1}}$, the law of large numbers implies almost surely that for any $\delta_{\ell} > 0$ sufficiently large sample size $n_{\ell}$ guarantees
\begin{equation}
\max\left\{ \norm{\hat{\bSigma}^{(\ell)}_{n_{\ell}} - \bSigma^{(\ell)}}_{\mathrm{F}} 
,\ \norm{\hat{\bSigma}^{(\ell)}_{n_{\ell-1}} - \bSigma^{(\ell)}}_{\mathrm{F}} \right\} 
<
\delta_{\ell} \, ,
\label{eq:Prop3.2Eq3}
\end{equation}
for $\ell = 0,\ldots,L$.  Setting $\delta_{\ell} = h/4$ and applying the triangle inequality again with the result of~\eqref{eq:Prop3.2Eq2} gives
\begin{multline}
\max\left\{ \norm{\hat{\bSigma}^{(\ell)}_{n_{\ell}} - \eye}_{\mathrm{F}} 
,\ \norm{\hat{\bSigma}^{(\ell)}_{n_{\ell-1}} - \eye}_{\mathrm{F}} \right\}  < \\
\delta_{\ell} + \norm{\bSigma^{(\ell)} - \eye}_{\mathrm{F}} 
<
\delta_{\ell} + \frac{h}{2} < \frac{3h}{4} < 1\, .
\label{eq:Prop3.2Eq4}
\end{multline}

Expanding the Taylor series for each matrix logarithm up to first-order (i.e. $\Log(\bSigma) \approx \bSigma - \eye$) and canceling out the resulting identity matrices gives
\begin{multline}
\norm{\Log \hat{\bSigma}_{\bn}^{\mathrm{LEMF}} -  \Log \bSigma}_{\mathrm{F}}
=\\
    \norm{ \hat{\bSigma}_{n_0}^{(0)} + \sum_{\ell=1}^L \alpha_{\ell} \left( \hat{\bSigma}^{(\ell)}_{n_{\ell}} - \hat{\bSigma}^{(\ell)}_{n_{\ell-1}} \right)  - \bSigma  + {\mathbf{M}} }_{\mathrm{F}} \, ,
\label{eq:Prop3.2Eq5}
\end{multline}
where the matrix ${\mathbf{M}}$ is the remainder of the first-order Taylor expansions
\begin{equation}
\begin{split}
    {\mathbf{M}} =& \sum_{k=2}^{\infty} \frac{(-1)^{k+1}}{k}(\hat{\bSigma}_{n_0}^{(0)} - \eye)^k \\
    &+ \sum_{\ell=1}^L \alpha_{\ell}\left( \sum_{k=2}^{\infty} \frac{(-1)^{k+1}}{k}(\hat{\bSigma}_{n_{\ell}}^{(\ell)} - \eye)^k \right)\\
    &- 
    \sum_{\ell=1}^L \alpha_{\ell}\left(
    \sum_{k=2}^{\infty} \frac{(-1)^{k+1}}{k}(\hat{\bSigma}_{n_{\ell}-1}^{(\ell)} - \eye)^k\right) \\
    &- 
    \sum_{k=2}^{\infty} \frac{(-1)^{k+1}}{k}(\bSigma - \eye)^k \, .
\end{split}
\label{eq:Prop3.2Eq6}
\end{equation}
Applying both the triangle and reverse-triangle inequalities gives
\begin{equation}
\left | 
\norm{\Log \hat{\bSigma}_{\bn}^{\mathrm{LEMF}} -  \Log \bSigma}_{\mathrm{F}} 
-
\norm{\hat{\bSigma}_{\bn}^{\mathrm{EMF}} -   \bSigma}_{\mathrm{F}} 
\right|
\le \norm{{\mathbf{M}}}_{\mathrm{F}}\, .
\label{eq:Prop3.2Eq7}
\end{equation}
To bound the right-hand-side of~\eqref{eq:Prop3.2Eq7}, we again use the triangle inequality and that $\norm{{\mathbf{A}}^k}_{\mathrm{F}} \le \norm{\mathbf{A}}_{\mathrm{F}}^k$ to obtain
\begin{equation}
\begin{split}
    \norm{{\mathbf{M}}}_{F} & \le  \sum_{k=2}^{\infty} \frac{1}{k}\norm{\hat{\bSigma}_{n_0}^{(0)} - \eye}_{\mathrm{F}}^k \\
    &+ \sum_{\ell=1}^L \alpha_{\ell}\left(
    \sum_{k=2}^{\infty} \frac{1}{k}\norm{\hat{\bSigma}_{n_{\ell}}^{(\ell)} - \eye}_{\mathrm{F}}^k\right)\\
    &+ 
    \sum_{\ell=1}^L \alpha_{\ell}\left(
    \sum_{k=2}^{\infty} \frac{1}{k}\norm{\hat{\bSigma}_{n_{\ell-1}}^{(\ell)} - \eye}_{\mathrm{F}}^k\right) \\
    &+ 
    \sum_{k=2}^{\infty} \frac{1}{k} \norm{\bSigma - \eye}_{\mathrm{F}}^k \, .
\end{split}
\label{eq:Prop3.2Eq8}
\end{equation}
By~\eqref{eq:Prop3.2Eq4} and the assumption $\norm{\bSigma - \eye}_{\mathrm{F}} < h/4$, each series in~\eqref{eq:Prop3.2Eq8} converges absolutely.  Bounding $1/k$ with $1$ and factoring out a quadratic power for each term in~\eqref{eq:Prop3.2Eq8} gives
\begin{equation}
\begin{split}
    \norm{{\mathbf{M}}}_{F} & \le  \norm{\hat{\bSigma}_{n_0}^{(0)} - \eye}_{\mathrm{F}}^2 \sum_{k=0}^{\infty} \norm{\hat{\bSigma}_{n_0}^{(0)} - \eye}_{\mathrm{F}}^k \\
    &+ \sum_{\ell=1}^L \alpha_{\ell}  \norm{\hat{\bSigma}_{n_{\ell}}^{(\ell)} - \eye}_{\mathrm{F}}^2  \left(
    \sum_{k=0}^{\infty} \norm{\hat{\bSigma}_{n_{\ell}}^{(\ell)} - \eye}_{\mathrm{F}}^k\right)\\
    &+ 
    \sum_{\ell=1}^L \alpha_{\ell} \norm{\hat{\bSigma}_{n_{\ell-1}}^{(\ell)} - \eye}_{\mathrm{F}}^2   \left(
    \sum_{k=0}^{\infty} \norm{\hat{\bSigma}_{n_{\ell-1}}^{(\ell)} - \eye}_{\mathrm{F}}^k\right) \\
    &+ 
    \norm{\bSigma - \eye}_{\mathrm{F}}^2 \sum_{k=0}^{\infty} \norm{\bSigma - \eye}_{\mathrm{F}}^k \, .
\end{split}
\label{eq:Prop3.2Eq9}
\end{equation}
Evaluating each of these geometric series gives the bound
\begin{equation}
\begin{split}
    \norm{{\mathbf{M}}}_{\mathrm{F}} & \le  \frac{\norm{\hat{\bSigma}_{n_0}^{(0)} - \eye}_{\mathrm{F}}^2}{1 - \norm{\hat{\bSigma}_{n_0}^{(0)} - \eye}_{\mathrm{F}}} \\
    &+ \sum_{\ell=1}^L \alpha_{\ell}\frac{\norm{\hat{\bSigma}^{(\ell)}_{n_{\ell}} - \eye}_{\mathrm{F}}^2}{1 - \norm{\hat{\bSigma}_{n_{\ell}}^{(\ell)} - \eye}_{\mathrm{F}}}\\
    &+ 
    \sum_{\ell=1}^L \alpha_{\ell}\frac{\norm{\hat{\bSigma}^{(\ell)}_{n_{\ell-1}} - \eye}_{\mathrm{F}}^2}{1 - \norm{\hat{\bSigma}_{n_{\ell-1}}^{(\ell)} - \eye}_{\mathrm{F}}} \\
    &+ 
    \frac{\norm{\bSigma - \eye}_{\mathrm{F}}^2}{1 - \norm{\bSigma - \eye}_{\mathrm{F}}}   \, .
\end{split}
\label{eq:Prop3.2Eq10}
\end{equation}
By~\eqref{eq:Prop3.2Eq4} we have that
\begin{multline}
    \max\left\{ 
    \frac{1}{1 - \norm{\hat{\bSigma}^{(\ell)}_{n_{\ell}} - \eye}_{\mathrm{F}}},\ 
    \frac{1}{1 - \norm{\hat{\bSigma}^{(\ell)}_{n_{\ell-1}} - \eye}_{\mathrm{F}}}
    \right\}
    < \\
    \frac{1}{1 - 3h/4} < 4 \, ,
\label{eq:Prop3.2Eq11}
\end{multline}
and similarly
\begin{equation}
 \frac{1}{1 - \norm{\bSigma - \eye}_{\mathrm{F}}} < \frac{1}{1 - h/4} < \frac{4}{3} \, .
\label{eq:Prop3.2Eq12}
\end{equation}
Substituting the bounds~\eqref{eq:Prop3.2Eq11} and~\eqref{eq:Prop3.2Eq12} into~\eqref{eq:Prop3.2Eq10} gives
\begin{equation}
\begin{split}
    \norm{{\mathbf{M}}}_{\mathrm{F}} & \le  4 \norm{\hat{\bSigma}_{n_0}^{(0)} - \eye}_{\mathrm{F}}^2 \\
    &+ 4\sum_{\ell=1}^L \alpha_{\ell} \norm{\hat{\bSigma}^{(\ell)}_{n_{\ell}} - \eye}_{\mathrm{F}}^2\\
    &+ 
    4\sum_{\ell=1}^L \alpha_{\ell} \norm{\hat{\bSigma}^{(\ell)}_{n_{\ell-1}} - \eye}_{\mathrm{F}}^2\\
    &+ 
    \frac{4}{3}\norm{\bSigma - \eye}_{\mathrm{F}}^2 \, .
\end{split}
\label{eq:Prop3.2Eq13}
\end{equation}
Applying the bound~\eqref{eq:Prop3.2Eq4} with~\eqref{eq:Prop3.2Eq13} gives
\begin{multline}
    \norm{ {\mathbf{M}}}_{\mathrm{F}} 
    \le
    4 \left( \frac{3h}{4} \right)^2 + 8\sum_{\ell=1}^L \alpha_{\ell} \left( \frac{3h}{4} \right)^2  + \frac{4}{3}\left( \frac{h}{4} \right)^2 \\
    \le
    \left(\frac{7}{3}  + \frac{9}{2}\sum_{\ell=1}^L \alpha_{\ell} \right) h^2
    = K h^2 \, .
\label{eq:Prop3.2Eq14}
\end{multline}
Factoring using the difference of squares formula ($a^2 - b^2 = (a-b)(a+b)$) gives
\begin{equation}
\begin{split}
    &\left| \norm{\Log \hat{\bSigma}_{\bn}^{\mathrm{LEMF}} -  \Log \bSigma}_{\mathrm{F}}^2 
- 
\norm{\hat{\bSigma}_{\bn}^{\mathrm{EMF}} -   \bSigma}_{\mathrm{F}}^2 
\right| = \\
& \left| \norm{\Log \hat{\bSigma}_{\bn}^{\mathrm{LEMF}} -  \Log \bSigma}_{\mathrm{F}}
+
\norm{\hat{\bSigma}_{\bn}^{\mathrm{EMF}} -   \bSigma}_{\mathrm{F}}  \right| \\
& \times \left| \norm{\Log \hat{\bSigma}_{\bn}^{\mathrm{LEMF}} -  \Log \bSigma}_{\mathrm{F}}
-
\norm{\hat{\bSigma}_{\bn}^{\mathrm{EMF}} -   \bSigma}_{\mathrm{F}}  \right|
\, ,
\end{split}
\label{eq:Prop3.2Eq15}
\end{equation}
Applying the law of large numbers again, we know that almost surely for sufficiently large sample sizes
\begin{equation}
\max\left\{ \norm{\hat{\bSigma}_{\bn}^{\mathrm{EMF}} -   \bSigma}_{\mathrm{F}},\ \norm{\Log \hat{\bSigma}_{\bn}^{\mathrm{LEMF}} -  \Log \bSigma}_{\mathrm{F}}\right\} < 1 \, .
\label{eq:Prop3.2Eq16}
\end{equation}
Combining the bound~\eqref{eq:Prop3.2Eq14} with~\eqref{eq:Prop3.2Eq7} and substituting into~\eqref{eq:Prop3.2Eq15} gives
\begin{equation}
\left| \norm{\Log \hat{\bSigma}_{\bn}^{\mathrm{LEMF}} -  \Log \bSigma}_{\mathrm{F}}^2 
- 
\norm{\hat{\bSigma}_{\bn}^{\mathrm{EMF}} -   \bSigma}_{\mathrm{F}}^2 
\right| \le 2K h^2 \, .
\end{equation}
Taking the expectation and applying Jensen's inequality gives the result.
\begin{multline}
\left| \mathbb{E}\left[ \norm{\Log \hat{\bSigma}_{\bn}^{\mathrm{LEMF}} -  \Log \bSigma}_{\mathrm{F}}^2 \right]
- 
\mathbb{E}\left[ \norm{\hat{\bSigma}_{\bn}^{\mathrm{EMF}} -   \bSigma}_{\mathrm{F}}^2 \right]
\right| \\
\le \mathbb{E}\left[ \left| \norm{\Log \hat{\bSigma}_{\bn}^{\mathrm{LEMF}} -  \Log \bSigma}_{\mathrm{F}}^2 
- 
\norm{\hat{\bSigma}_{\bn}^{\mathrm{EMF}} -   \bSigma}_{\mathrm{F}}^2 
\right| \right]\\
\le 2K h^2 \, .
\end{multline}
\end{proof}

\section{Proof of \Cref{thm:OptimalAllocation} and \Cref{cor:MSEComparison}}
Since the Euclidean MSE (\cref{prop:MSELCV}) has the exact same form as the MSE of the scalar multi-fidelity Monte Carlo estimator in \citet{PWG16}, the optimal allocations and weights there apply directly to the Euclidean multi-fidelity covariance estimator \eqref{eq:LCV} and in first-order to the LEMF estimator \eqref{eq:TS}, with the exception that our definitions of $\sigma^2_\ell$ and $\rho_\ell$ are generalized for matrices. Equation~\ref{eq:WhenBetterBound} directly follows from \eqref{eq:FirstOrderMSELEMF}.

\begin{figure*}[!t]
\begin{lstlisting}[language=Python,frame = single]
from numpy import cov
from scipy.linalg import logm, expm

# compute sample-size vector n[] and coefficients a[] as in Proposition 3.4
# collect samples n_l x d of level l = 0, ..., L in array samps[l]

S = logm(cov(samps[0].T))
for i in range(len(a)):
  S = S + a[i]*(logm(cov(samps[i+1].T))-logm(cov(samps[i+1][:n[i]].T)))
S = expm(S)
\end{lstlisting}
\caption{This code sketches an implementation of the LEMF estimator. It builds on standard numerical linear algebra routines that are readily available in NumPy/SciPy \citep{harris2020array,2020SciPy-NMeth}.}
\label{fig:CodeListing}
\end{figure*}
\section{Algorithmic Description and Code Sketch of LEMF Estimator}\label{sec:Appdx:Code}
The code shown in Figure~\ref{fig:CodeListing} sketches an implementation of the LEMF estimator. It is important to note that it only requires numerical linear algebra functions that are readily available in, e.g., NumPy/SciPy \citep{harris2020array,2020SciPy-NMeth}. A detailed implementation to reproduce the shown numerical results is provided in the supplemental material.

\section{Supplemental to ``Motivating Example with Gaussian''}
\label{appx:ToyGaussian}

Consider high-fidelity samples $\by^{(0)} \sim N(\zeros, \bSigma)$ with $\bSigma \in \mathbb{S}^d_{++}$ and low-fidelity samples which correspond to high-fidelity samples corrupted by additional independent noise,
\begin{equation}
    \by^{(\ell)} = \by + \beps^{(\ell)}\, ,\qquad \beps \sim N(\zeros, \bGamma^{(\ell)}) \, .
\label{eq:ToyGaussianLFSample}
\end{equation}
Because the samples are Gaussian we may compute explicitly the generalized variances~\eqref{eq:GeneralizedVariance} and generalized correlations~\eqref{eq:GeneralizedCorrelation} needed for the optimal sample allocation in Theorem~\ref{thm:OptimalAllocation}.

\paragraph{Generalized variance}
For the generalized variance we have that
\begin{align*}
    &\tr(\cov[\by^{(0)}(\by^{(0)})^{\top}]) \\
    &= \mathbb{E}\left[ \norm{\by^{(0)} (\by^{(0)})^{\top} - \bSigma}^2_{\mathrm{F}} \right] \\
    &= 
    \mathbb{E}\left[ \norm{\by^{(0)} (\by^{(0)})^{\top}}^2_{\mathrm{F}} \right] 
    -2 \mathbb{E}\left[\inner{ \by^{(0)} (\by^{(0)})^{\top} }{\bSigma}_{\mathrm{F}}\right]
    + \norm{\bSigma}_{\mathrm{F}}^2 \\
    &= 
    \mathbb{E}\left[ \tr(\by^{(0)} (\by^{(0)})^{\top}\by^{(0)} (\by^{(0)})^{\top}) \right] - \norm{\bSigma}^2_{\mathrm{F}} \, .
\end{align*}

Focusing on the first term we see
\begin{multline*}
\tr\left(\norm{\by^{(0)}}^2 \by^{(0)} (\by^{(0)})^{\top} \right) 
= \norm{\by^{(0)}}^2 \sum_{i=1}^d \left(y_i^{(0)} \right)^2 
= \norm{\by^{(0)}}^4 \, ,
\end{multline*}
so that
\[
    \tr\left(\cov[\by^{(0)}(\by^{(0)})^{\top}]\right) = \mathbb{E}\left[ \norm{\by^{(0)}}^4 \right] - \norm{\bSigma}^2_{\mathrm{F}} \, ,
\]
and is analogous to the formula $\mathrm{Var}[x] = \mathbb{E}[x^2] - \mathbb{E}[x]^2$ for scalar random variables.  For a multivariate Gaussian random variable with zero mean we have that
\begin{equation}
    \mathbb{E}\left[ \norm{\by^{(0)}}^4 \right] = \tr(\bSigma)^2 + 2\tr(\bSigma^2) \, .
\label{eq:FourthMomentGaussian}
\end{equation}
Therefore,
\begin{equation}
    \tr(\cov[\by^{(0)}(\by^{(0)})^{\top}]) = \tr(\bSigma^2) + \tr(\bSigma)^2 \, .
\label{eq:ToyGaussianGeneralizedVariance}
\end{equation}
Similarly, we may apply~\eqref{eq:ToyGaussianGeneralizedVariance} for the low-fidelity samples $\by^{(\ell)} \sim N(\zeros, \bSigma + \bGamma^{(\ell)})$.

\paragraph{Generalized correlation}  In order to compute the generalized correlation we compute the cross-covariance,
\begin{align*}
    &\tr(\cov[\by^{(0)} (\by^{(0)})^{\top},\ (\by^{(\ell)})(\by^{(\ell)})^{\top}]) \\
    &=
    \mathbb{E}\left[ \inner{\by^{(0)}(\by^{(0)})^{\top} - \bSigma}{(\by^{(\ell)})(\by^{(\ell)})^{\top} - \bSigma^{(\ell)}}_{\mathrm{F}} \right] \\
    &= 
    \mathbb{E}\left[ \tr\left( \by^{(0)}(\by^{(0)})^{\top}(\by^{(\ell)})(\by^{(\ell)})^{\top} \right) \right] - \tr(\bSigma \bSigma^{(\ell)}) \\
    &= 
    \mathbb{E}\left[ \left( (\by^{(0)})^{\top}(\by^{(\ell)})\right)^2 \right] - \tr(\bSigma \bSigma^{(\ell)})
    \, ,
\end{align*}
where we have used the notation $\bSigma^{(\ell)} = \bSigma + \bGamma^{(\ell)}$. 
Writing $\by^{(\ell)} = \by^{(0)} + \beps^{(\ell)}$ and expanding the squared inner product gives
\begin{align*}
    \left( (\by^{(0)})^{\top}(\by^{(\ell)})\right)^2 
    &=
    \left(\norm{\by^{(0)}}^2 + (\by^{(0)})^{\top}\beps^{(\ell)} \right)^2\\
    &= \norm{\by^{(0)}}^4 + 2\norm{\by^{(0)}}^2 (\by^{(0)})^{\top}\beps^{(\ell)} + \\ & \qquad \qquad \qquad \qquad \quad \left( (\by^{(0)})^{\top}\beps^{(\ell)}\right)^2 \, .
\end{align*}
For the first term we may use~\eqref{eq:FourthMomentGaussian} while for the second term we have by independence that
\begin{multline*} 
    \mathbb{E}\left[\norm{\by^{(0)}}^2 (\by^{(0)})^{\top}\beps^{(\ell)} \right] 
    = \\
    \mathbb{E}\left[\norm{\by^{(0)}}^2 \by^{(0)} \right]^{\top} \mathbb{E}\left[ \beps^{(\ell)}\right]
    = 0 \, .
\end{multline*} 
For the last term we rely on properties of the trace operator.
\begin{align*}
\mathbb{E}\left[ \left( (\by^{(0)})^{\top}\beps^{(\ell)}\right)^2 \right]
&= 
\mathbb{E}\left[ \tr\left( (\by^{(0)})^{\top}\beps^{(\ell)}(\by^{(0)})^{\top}\beps^{(\ell)} \right) \right] \\
&= 
\mathbb{E}\left[ \tr\left( (\by^{(0)})^{\top}\beps^{(\ell)}(\beps^{(\ell)})^{\top} \by^{(0)} \right) \right] \\
&=
\mathbb{E}\left[ \tr\left( \by^{(0)}(\by^{(0)})^{\top} (\beps^{(\ell)})(\beps^{(\ell)})^{\top} \right) \right] \\
&=
\tr\left( \mathbb{E}[\by^{(0)}(\by^{(0)})^{\top}] \mathbb{E}[ (\beps^{(\ell)})(\beps^{(\ell)})^{\top} ] \right) \\
&= 
\tr(\bSigma \bGamma^{(\ell)})\, .
\end{align*}
Combining everything gives the cross-covariance 
\begin{multline}
    \tr(\cov[\by^{(0)} (\by^{(0)})^{\top},\ (\by^{(\ell)})(\by^{(\ell)})^{\top}]) \\
    =
    \tr(\bSigma)^2 + \tr(\bSigma^2) = \sigma_0^2 \, ,
\label{eq:ToyGaussianGeneralizedCrossCov}
\end{multline}
which is analogous to the scalar setting where $\cov[y,\epsilon] = 0$ so that $\cov[y,y+\epsilon] = \mathrm{Var}[y]$.  Using the notation $\sigma_0,\sigma_{\ell}$, and $\rho_{\ell}$ defined in Section~\ref{sec:3.3} we have that the generalized correlation is given by
\begin{equation}
    \rho_{\ell} = \frac{\sigma_0}{\sigma_{\ell}} \, .
\end{equation}
Notice that $\sigma_{\ell} > \sigma$ and so $\rho_{\ell} < 1$.

\paragraph{Example set up}

For the example presented in Section~\ref{sec:NumExp:Gaussian}, we select $\bSigma$ randomly from the $\text{Wishart}(4)$ ensemble, that is, $\bSigma = {\mathbf{A}}^{\top}{\mathbf{A}}$ where ${\mathbf{A}} \in \mathbb{R}^{4\times 4}$ has i.i.d. standard normal entries. The particular realization of $\bSigma$ we obtained was
\[
    \bSigma = \begin{bmatrix}
2.52 &-0.17 &0.67 &-0.98\\
-0.17 &0.64 &-0.29 &0.35\\
0.67 &-0.29 &0.49 &-0.52\\
-0.98 &0.35 &-0.52 &1.31\\
\end{bmatrix} \, .
\]
The noise covariance matrices defining the low-fidelity samples at each level $\ell \in \{1,2,3\}$ are given by
\[
    \bGamma^{(1)} = 0.1 \eye_{4\times 4},\quad \bGamma^{(2)} = 0.5\eye_{4\times 4},\quad \bGamma^{(3)} = \eye_{4\times 4} \, ,
\]
and result in generalized correlations~\eqref{eq:GeneralizedCorrelation}
\[
    \rho_1 \approx 0.93,\quad \rho_2 \approx 0.74,\quad \rho_3 \approx 0.58\, .
\]
We impose sampling costs $c_0 = 1$ and $c_{\ell} = 10^{-\ell-1}$ for $\ell\in \{1,2,3\}$. Setting the computational budget to $B=15$ and taking sample sizes $\left\lfloor \bn^* \right\rfloor$ from Theorem~\ref{thm:OptimalAllocation} yields
\[
    n =  12,\quad  n_1 = 199,\quad  n_2 = 505,\quad n_3 = 2073 \, .
\]
Using the derived sample allocation we draw the corresponding number of samples from each level to obtain the three multi-fidelity estimators: EMF~\eqref{eq:LCV}, LEMF~\eqref{eq:TS}, and the truncated estimator~\eqref{eq:TruncEig}.  
For the truncated multi-fidelity estimator, we set a threshold of $\delta = 10^{-16}$ on the eigenvalues.  
Additionally, we also draw $n_{\mathrm{high}} = B/c  =15$ high-fidelity samples and $n_{\mathrm{low}} = B/c_3 = 1.5 \times 10^{5}$ low-fidelity samples $\by^{(3)}$ to compute single-fidelity estimates with the same computational budget.
The results reported in Table~\ref{table:ToyGaussianMSE} are estimated MSE averaged over 100 independent trials. Across all 100 trials, about 5\% of the EMF estimates were indefinite.

\section{Supplemental to ``Uncertainty Quantification of Steady-State Heat Flow''}\label{appx:UQ}
For demonstration purposes, we show in Figure~\ref{fig:HeatSol} the numerical solution of the steady-state heat-flow problem \eqref{eq:heat_pde} for three different, randomly sampled parameters $\btheta$. Observations $\by^{(\ell)}$, $\ell \in \{0, 1\}$, consist of temperature measurements at 10 equidistant locations over the spatial domain, where $\ell$ corresponds to the number of grid points used to numerically solve \eqref{eq:heat_pde}. Our goal in this example is to estimate the covariance of $\by^{(0)}$, the observations corresponding to the highest number of solver grid points.

We perform a pilot study to estimate the generalized variances and correlations $\sigma_0,\sigma_1,\rho_1$ 
needed to compute the optimal multi-fidelity sample allocations $\bn^{\star} = [n_{\star}^{(0)}, n_{\star}^{(1)}]$ and coefficients $\balpha^{\star}$. We use $10^5$ samples of $\btheta$ for the pilot study to avoid mixing errors of the optimal sample allocation with the MSEs of the estimators. %

Figure~\ref{fig:fwd_uq_mse_extend} is analogus to Figure~\ref{fig:fwd_uq_mse} and shows the MSEs of our LEMF estimator compared to the MSEs of single-fidelity and other multi-fidelity estimators, except that now the minimum and maximum MSEs over all 100 trials are shown with error bars.

Figure~\ref{fig:NumExp:HeatFlowSpeedupSupp}a-b shows speedups of our LEMF estimator with respect to the Euclidean and affine-invariant metric. As the budget is increased, the Euclidean estimator becomes indefinite less frequently, which can be seen in Figure~\ref{fig:NumExp:HeatFlowSpeedupSupp}c.

Figure~\ref{fig:Heat:AllCompoments} visualizes the convergence behavior of various estimators. Our LEMF estimator converges more quickly to the true value of $\bSigma$ element-wise than the single-fidelity estimators; this is particularly true for the diagonal elements of $\bSigma$. The y-scale is comparable to that of Figure~\ref{fig:fwd_uq_other}c and the x-scale shows the costs and is the same as the x-axis of Figure~\ref{fig:fwd_uq_other}c.

\begin{figure}
    \centering
    \resizebox{0.49\textwidth}{!}{\Large \input{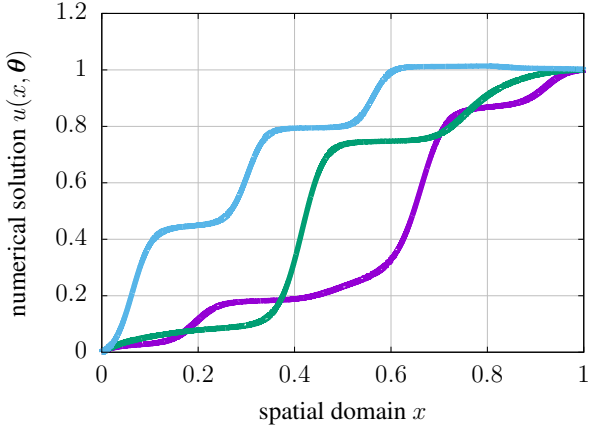}}
    \caption{Supplemental to heat flow example: Three sample solutions (temperature) of the steady-state heat equation~\eqref{eq:heat_pde} corresponding to different random parameters $\btheta \in \mathbb{R}^4$.}
    \label{fig:HeatSol}
\end{figure}

\begin{figure*}
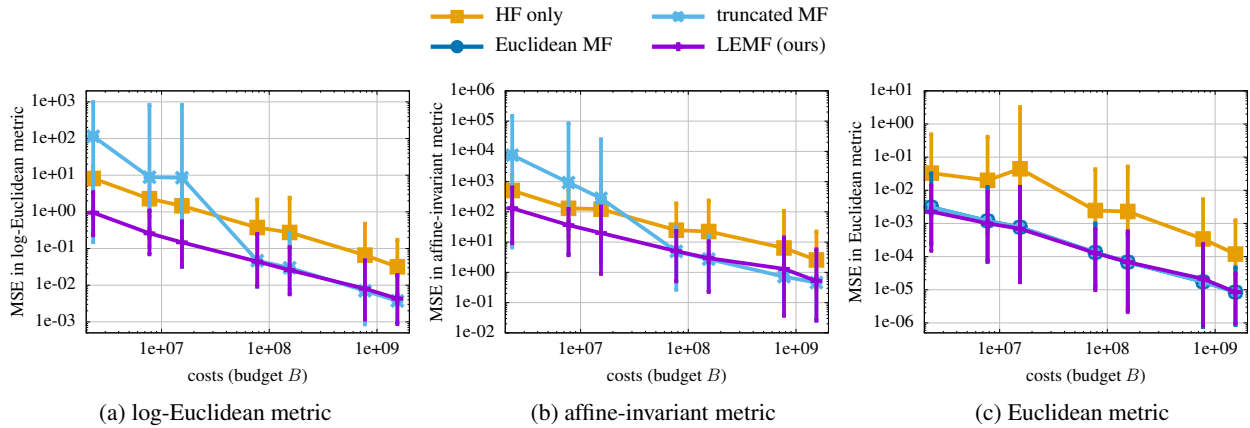

\begin{center}
\setlength\tabcolsep{1pt}
\begin{tabular}{ccc}
\multicolumn{3}{c}{
\resizebox{0.4\textwidth}{!}{\Large \input{figures/fwd_uq_mse_key_eb.tex}}}\vspace*{-2cm}\\
\resizebox{0.35\textwidth}{!}{\Large \input{figures/fwduq_heat_logE_eb.tex}}
&
\hspace*{-0.5cm}\resizebox{0.35\textwidth}{!}{\Large \input{figures/fwduq_heat_aff_eb.tex}}
& 
\hspace*{-0.5cm}\resizebox{0.35\textwidth}{!}{\Large \input{figures/fwduq_heat_frob_eb.tex}}\\
\small (a) log-Euclidean metric & \small (b) affine-invariant metric&
\small (c) Euclidean metric
\end{tabular}
\end{center}
\caption{Supplemental for heat-flow problem: Plots analogous to Figure~\ref{fig:fwd_uq_mse}, except that the minimum and maximum MSE over all 100 trials are shown as error bars. Note that the surrogate-only estimator is not shown to ease exposition. The Euclidean multi-fidelity estimator is shown only in plot (c) because it becomes indefinite and therefore the MSE cannot be computed in the log-Euclidean and the affine-invariant metric.}
\label{fig:fwd_uq_mse_extend}
\end{figure*}

\begin{figure*}
\begin{tabular}{ccc}
\resizebox{0.33\textwidth}{!}{\Large \input{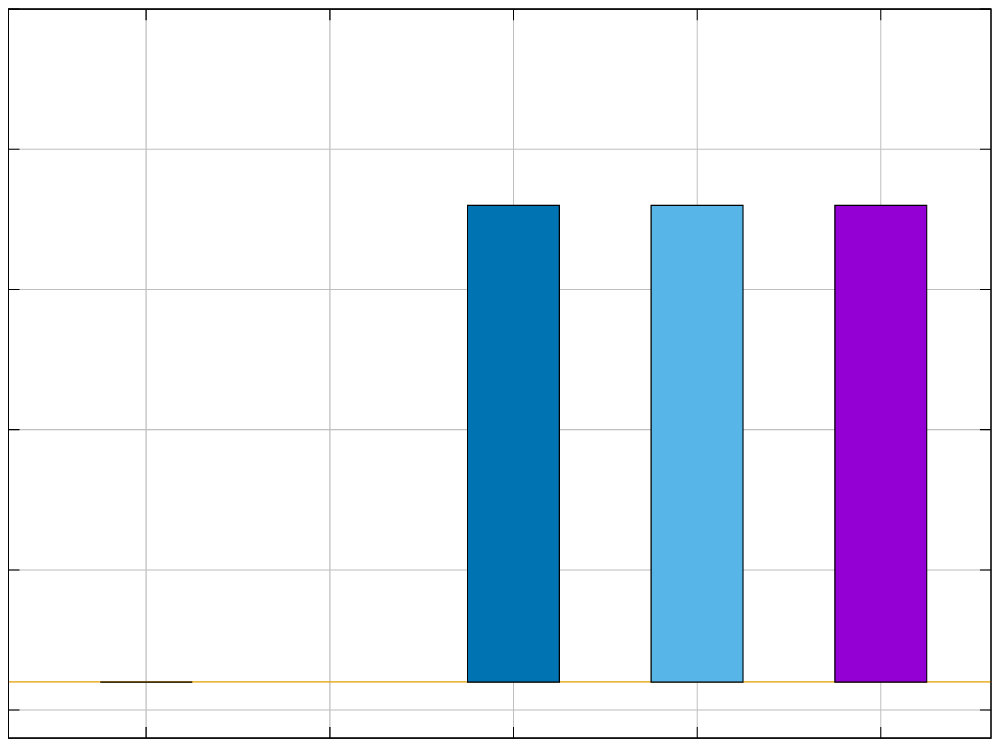}}
&
\hspace*{-0.75cm}\resizebox{0.33\textwidth}{!}{\Large \input{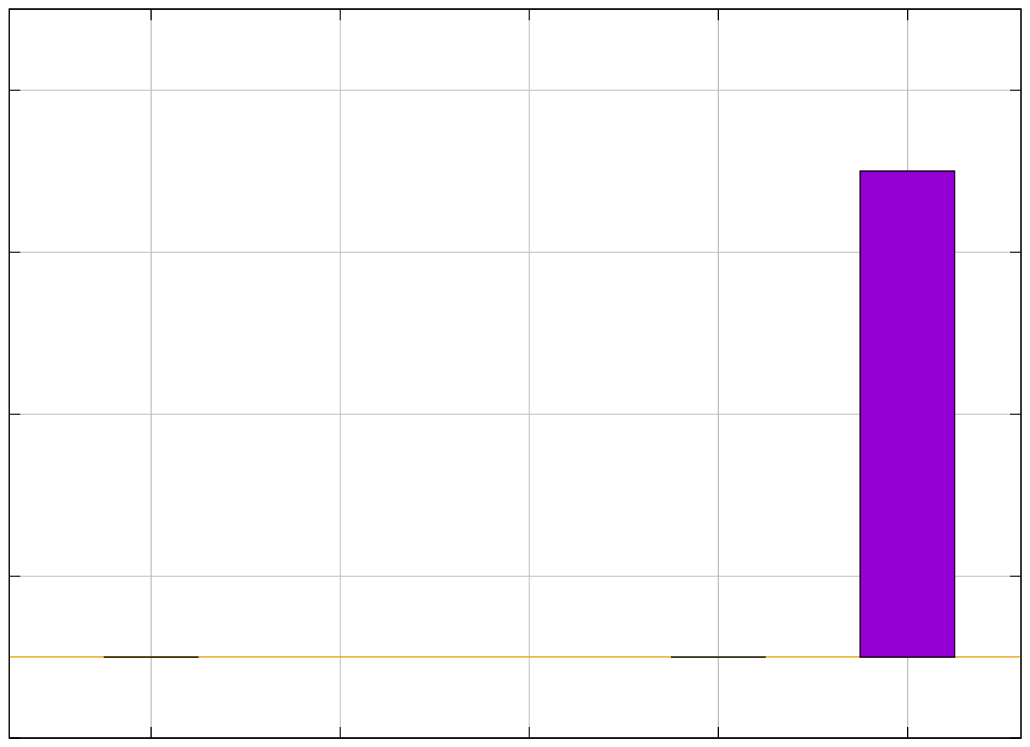}}
&
\hspace*{-0.75cm}\resizebox{0.33\textwidth}{!}{\Large \input{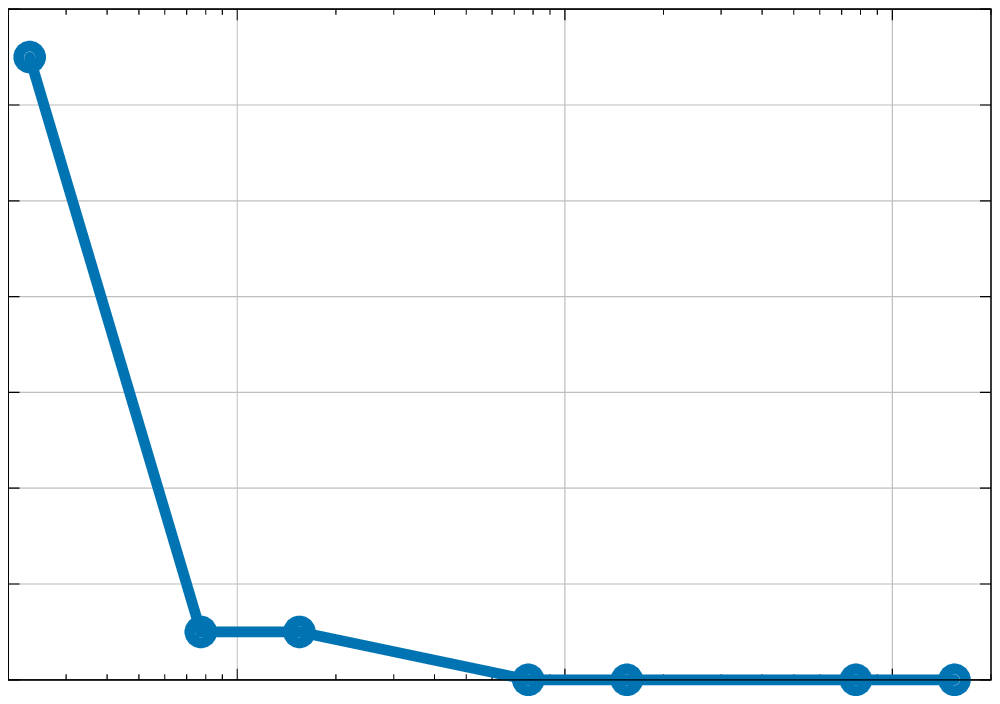}}\\
\small (a) speedup Euc.~metric, tol.~$2\times 10^{-4}$ & \hspace*{-0.75cm}\small (b) speedup affine-invariant metric, tol.~$3\times 10^1$ & \hspace*{-0.75cm}\small (c) indefiniteness of Euc.~MF
\end{tabular}
\caption{Supplemental for heat flow problem: Plots (a) and (b) show the speedups of our LEMF estimator compared to the single-fidelity estimator that uses the high-fidelity samples alone. Note that the Euclidean estimator is indefinite for more than 10\% of the trials and thus the speedup with respect to the affine-invariant metric cannot be computed. Plot (c) shows the percentage of the trails of the Euclidean multi-fidelity estimator that lead to an indefinite covariance matrix estimate.
}
\label{fig:NumExp:HeatFlowSpeedupSupp}
\end{figure*}

\begin{figure*}
\centering
\resizebox{1.0\textwidth}{!}{\tiny \input{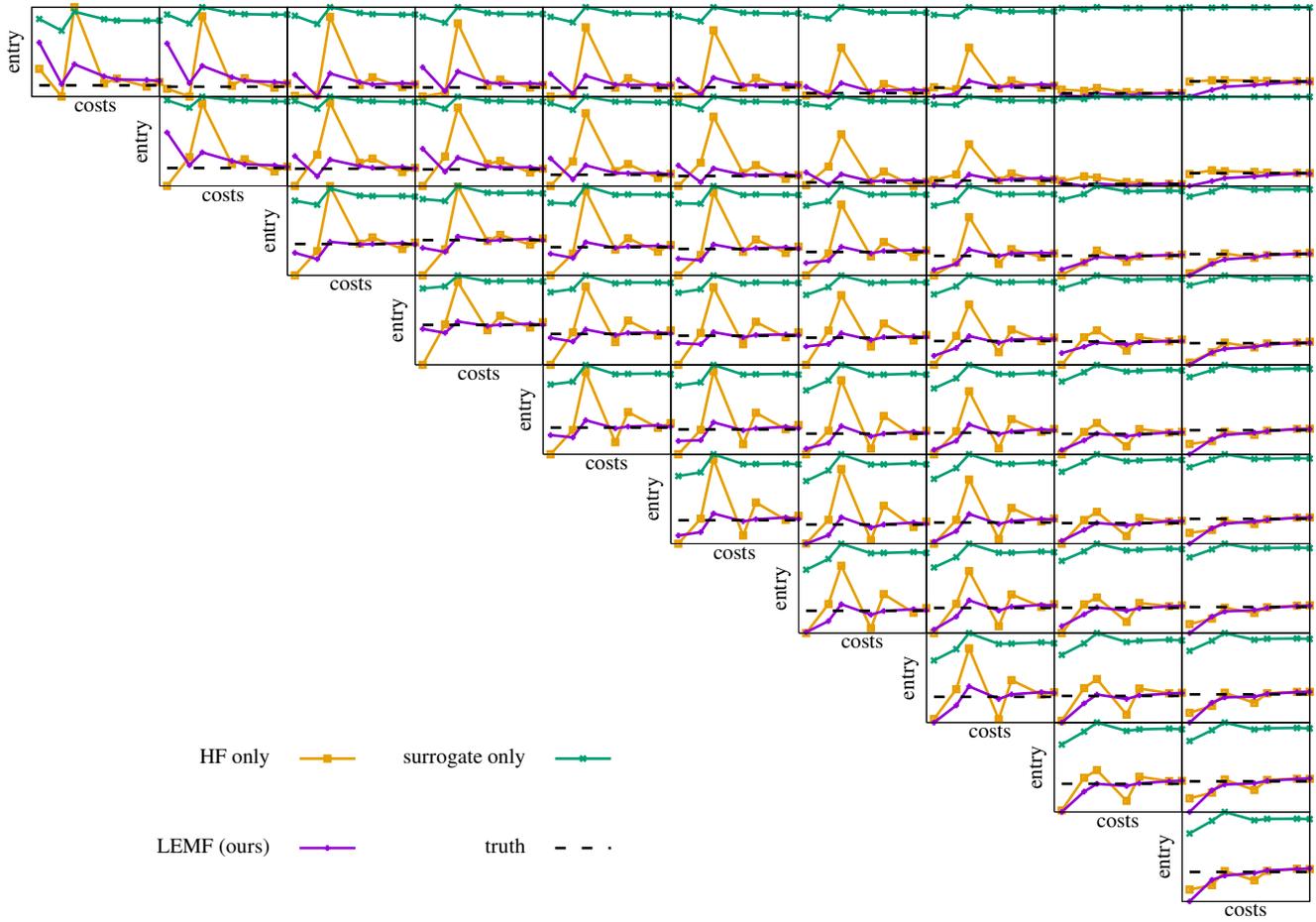}}
\caption{Supplemental for heat-flow problem: The plot shows the convergence behavior of the estimators. Our LEMF estimator is closer to the true value (dashed black) than the single-fidelity estimator, especially for low costs and for diagonal entries of the covariance matrix.}
\label{fig:Heat:AllCompoments}
\end{figure*}

\section{Supplemental to ``Metric Learning: Surface Quasi-Geostrophy''}
\label{appx:SQG}

We follow the surface quasi-geostrophic (SQG) model for the evolution of surface buoyancy $b:\mathcal{X} \times [0,\infty) \to \mathbb{R}$ in the periodic domain $\mathcal{X} = [-\pi,\pi]^2 \times (-\infty,0]$ presented in~\citet{HeldEtAl1985} and~\citet{CapetEtAl2008}.  
\begin{equation}
\begin{split}
\frac{\partial }{\partial t} b(\bx,t;\btheta) + J(\psi, b) &= 0\, ,\quad z = 0\, ,
\\
b &= \frac{\partial}{\partial z}\psi 
\\
\Delta \psi &= 0\, ,\quad z < 0 
\\
\psi &\to 0\, , \quad z \to -\infty \, .
\end{split}
\label{eq:SQGModel}
\end{equation}
Here $\bx = [x,y,z]^{\top} \in \mathcal{X}$ and the Jacobian is
\begin{equation}
    J(\psi, b) = \left( \frac{\partial \psi}{\partial x} \right)\left(\frac{\partial b}{\partial y}  \right) - \left( \frac{\partial b}{\partial x}  \right)\left(\frac{\partial \psi}{\partial y}  \right) \, .
\end{equation}
The SQG equation~\eqref{eq:SQGModel} is solved by applying a Fourier transform and integrating in time with finite differences.  The parameter vector $\btheta = [\theta_1,\ldots,\theta_5]^{\top}$ defines the initial buoyancy $b_0:\mathbb{R}^2 \to \mathbb{R}$ at the surface $z = 0$ as well as the flow.  The initial buoyancy is given by the Gaussian
\begin{equation}
    b_0(x,y;\btheta) = -\frac{1}{(2\pi/ |\theta_5|)^2} \exp\left( -x^2 - \exp(2\theta_1) y^2 \right) \, ,
\end{equation}
where $\theta_1$ is the log aspect ratio that determines the shape of the ellipse, see Figure~\ref{fig:SQG_solution}, and $\theta_5$ controls the amplitude.  
Of the remaining parameters $\theta_2$ is the gradient Coriolis parameter, $\theta_3$ is the log buoyancy frequency, and $\theta_4$ is background zonal flow along the x-axis. %

The high-fidelity random variable consists of nine equally-spaced observations of the solution $b(\bx, T; \btheta)$ with $T=24$, where $b(\cdot, \cdot \;; \btheta)$ is computed numerically using a time step of size $\Delta t = 0.005$ and 256 grid points along each coordinate axis.  
The surrogate random variable consists of observations of the solution at the same locations and time except with $b(\cdot, \cdot \; ; \btheta)$ computed numerically using only 64 grid points along each coordinate axis.
The costs of each sampling each random variable are taken to be proportional to the total number of spatial grid points used in the solvers, $c_0 = 256^2 = 65,536$ and $c_1 = 64^2 = 4,096$. The observation points are shown in plot (c) and (d) in Figure~\ref{fig:SQG_solution}.

The input parameters $\btheta$ are drawn from a Gaussian mixture model where $i \in \{0,1\}$ is a latent variable determining the mixture component  sampled from.  The full distribution of $\btheta$ is 
\[
    p(\btheta \mid i) = N({\boldsymbol{\mu}}_i, {\mathbf{C}}) \, ,\quad i \sim \mathrm{Bernoulli}(1/2)\, ,
\]
where
\[
    {\boldsymbol{\mu}}_0 = [1, 0, 0, 0, 4]^{\top}\, ,\quad 
    {\boldsymbol{\mu}}_1 = [0.1, 0, 0, 0, 4]^{\top} \, ,
\]
and
\[
    {\mathbf{C}} = \begin{bmatrix}
    0.3^2 & & & & \\
    & 0.003^2 & & & \\
    & & 0 & & \\
    & & & 0.08^2 & \\
    & & & & 0.3^2 
    \end{bmatrix} \, .
\]
Note that we fix the third parameter $\theta_3 = 0$, but the covariance matrix of the output observations $\by^{(0)}$ and $\by^{(1)}$ will still be positive definite.

The first class $i=0$ corresponds to an initial buoyancy with a larger log-aspect ratio, which introduces an instability, while the second class $i=1$ corresponds to a smaller log-aspect ratio.
Figure~\ref{fig:SQG_solution} shows the buoyancy at initial and final time for stochastic inputs corresponding to both classes. We sample the class variable $i$ from a Bernoulli distribution with parameter $1/2$, i.e., $i = 0$ and $i = 1$ are drawn with equal probability. %

We define the mean relative error of distances $\mathrm{MRE}({\mathbf{A}})$ for the metric ${\mathbf{A}}$ with respect to the reference metric ${\mathbf{A}}_0$ as
\[
        \mathrm{MRE}({\mathbf{A}}) =  \frac{1}{5000} \sum\nolimits_{i=1}^{5000}  \frac{ | d_{{\mathbf{A}}}(\by_i,\zeros) - d_{\mathbf{A}_0}(\by_i,\zeros) | }{d_{\mathbf{A}_0}(\by,\zeros)}   \, .
    \]

\begin{figure*}
\centering
\begin{tabular}{ccc}
\includegraphics[width=0.40\textwidth]{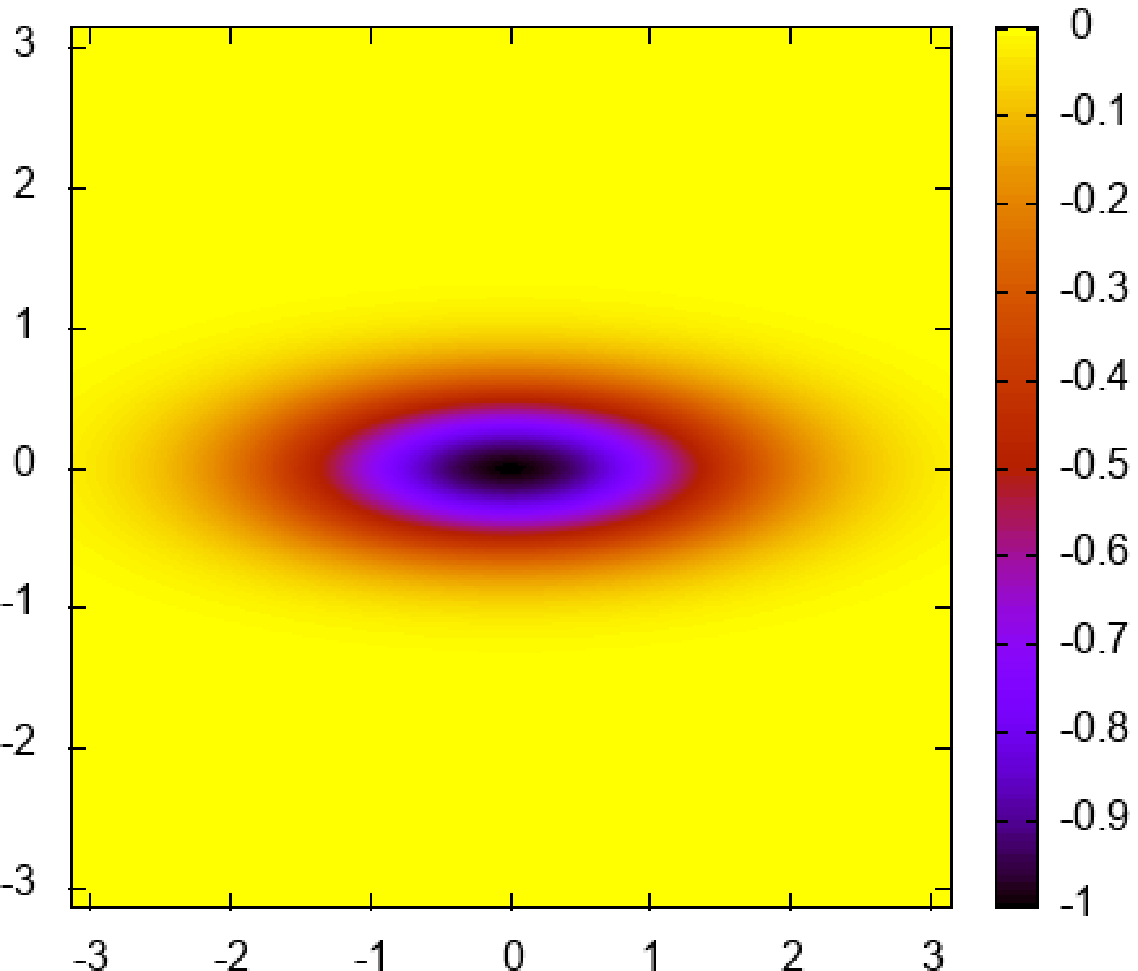}
&
~~~~
&
\includegraphics[width=0.40\textwidth]{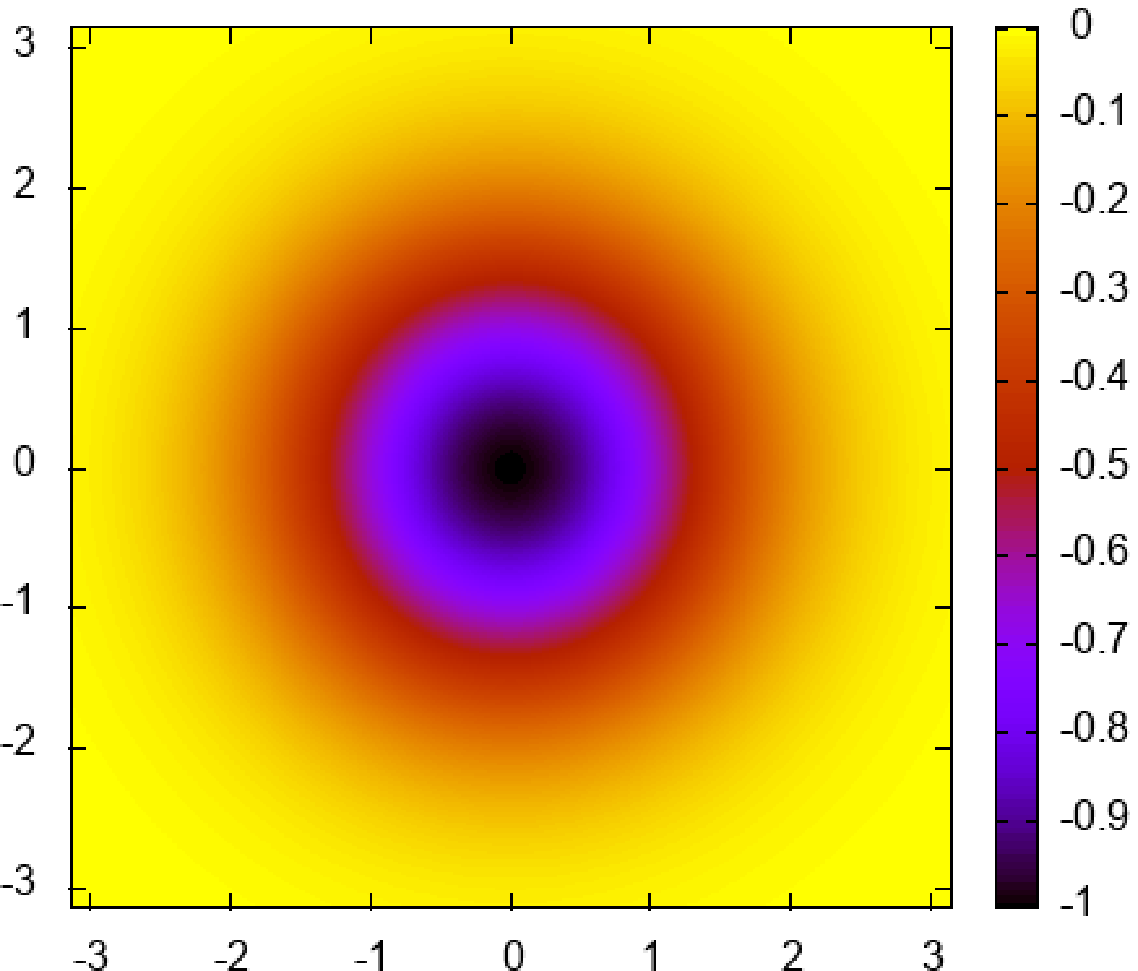}
\\
\small (a) class $i = 0$, initial buoyancy
&&
\small (b) class $i = 1$, initial buoyancy
\\
& \\
\includegraphics[width=0.40\textwidth]{figures/sqg_classA_final.eps}
&&
\includegraphics[width=0.40\textwidth]{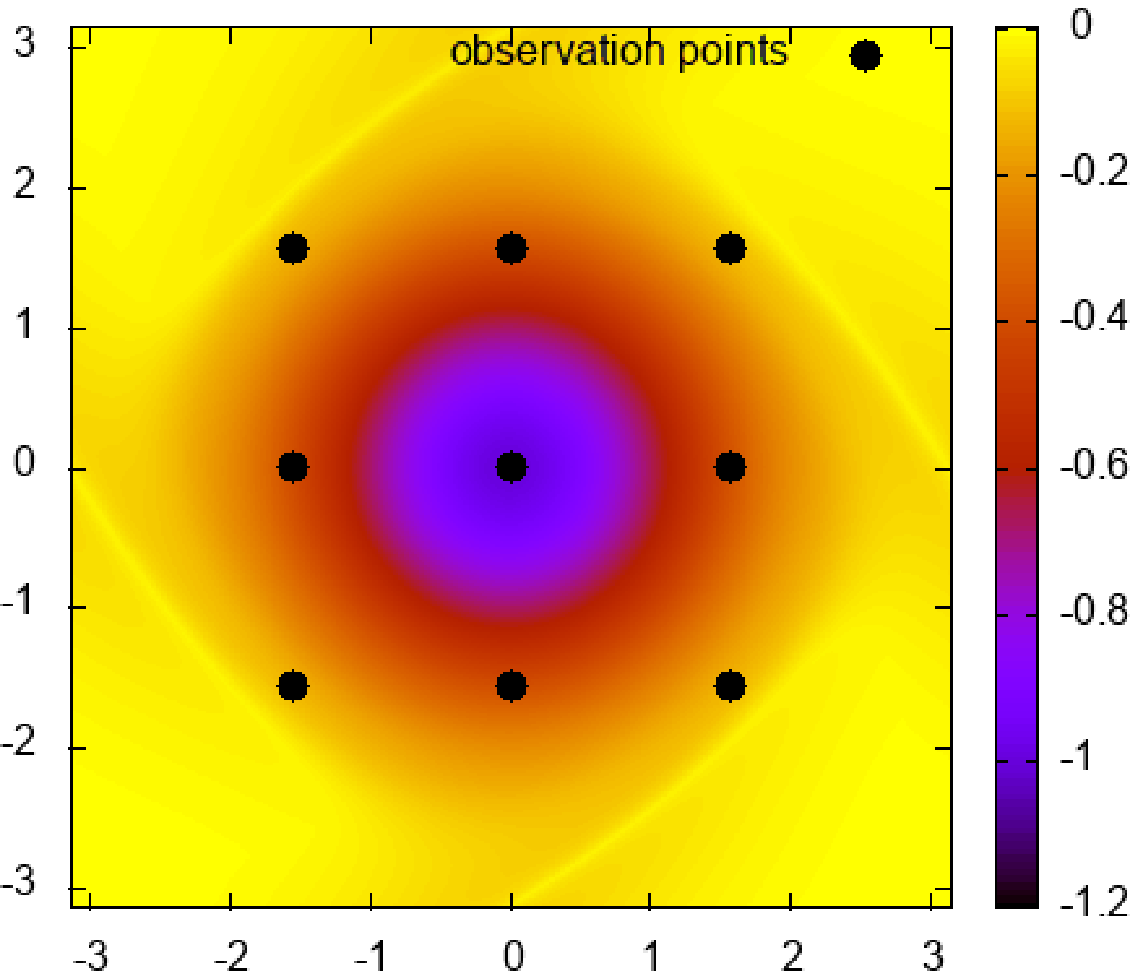}
\\
\small (c) class $i = 0$, final-time buoyancy
&&
\small (d) class $i = 1$, final-time buoyancy
\end{tabular}
\caption{Supplemental for metric learning: Plots show examples of the buoyancy at initial (top) and final time (bottom) for $\btheta$ sampled from class $i = 0$ (left) and $i = 1$ (right). Observations consist of solution values at nine spatial locations in the domain, as depicted in plots (c) and (d). We use the observations to estimate a metric which will distinguish between observations corresponding to $\btheta$ sampled from class $i = 0$ and $\btheta$ sampled from class $i = 1$. }
\label{fig:SQG_solution}
\end{figure*}

\end{document}